\documentclass[12pt]{article}
\usepackage{amsthm,amsfonts,amsmath, amscd, bbold}
\usepackage{mathrsfs}
\usepackage{accents}
\usepackage{authblk}

                                \usepackage{verbatim}
\swapnumbers
                              
\theoremstyle{plain}

\usepackage{enumitem}
\numberwithin{equation}{section}
\newtheorem{theorem}{Theorem}[section]

\newtheorem{lemma}[theorem]{Lemma}
\newtheorem{proposition}[theorem]{Proposition}
\newtheorem{corollary}[theorem]{Corollary}

\theoremstyle{definition}
\newtheorem{definition}[theorem]{Definition}
\newtheorem{example}[theorem]{Example}

\theoremstyle{remark}

\numberwithin{equation}{section}
\newcommand{\R}{{\mathbb R}}
\newcommand{\bR}{{\mathbb R}}

\newcommand{\cB}{{\mathcal B}}

\newcommand{\cH}{{\mathcal H}}

\newcommand{\cN}{{\mathcal N}}

\newcommand{\cK}{{\mathcal K}}
\newcommand{\sR}{{\mathord{\mathscr R}}}

\newcommand{\one}{{\bf 1}}

\newcommand{\ket}[1]{\left\vert #1\right\rangle}
\newcommand{\bra}[1]{\left\langle #1\right\vert}

\renewcommand{\Re}{\,\mathrm{Re}\,}   %{\mathfrak{Re}}
\renewcommand{\Im}{\,\mathrm{Im}\,}   %{\mathfrak{Im}}

\newcommand{\Tr}{\mathrm{Tr}}

\newcommand{\be}{\begin{equation}}
\newcommand{\ee}{\end{equation}}
\newcommand{\bea}{\begin{eqnarray}}
\newcommand{\eea}{\end{eqnarray}}
\newcommand{\beann}{\begin{eqnarray*}}
\newcommand{\eeann}{\end{eqnarray*}}

\usepackage{color}

\begin{document}

\title{On quantum quasi-relative entropy}
\author{Anna Vershynina}
\affil{\small{Department of Mathematics, Philip Guthrie Hoffman Hall, University of Houston, 
3551 Cullen Blvd., Houston, TX 77204-3008, USA}}
\renewcommand\Authands{ and }
\renewcommand\Affilfont{\itshape\small}

\date{\today}

\maketitle

\begin{abstract} We consider a quantum quasi-relative entropy $S_f^K$ for an operator $K$ and an operator convex function $f$. We show how to obtain the error bounds for  the monotonicity and joint convexity inequalities from the recent results for the $f$-divergences (i.e. $K=I$). We also provide an error term for a class of operator inequalities, that generalize operator strong subadditivity inequality. We apply those results to demonstrate explicit bounds for the logarithmic function, that leads to the quantum relative entropy, and the power function, which gives, in particular, a Wigner-Yanase-Dyson skew information. In particular, we provide the remainder terms for the  strong subadditivity inequality, operator strong subadditivity inequality, WYD-type inequalities, and the Cauchy-Schwartz inequality.
\end{abstract}

\tableofcontents

\section{Introduction}

Quantum quasi-relative entropy was introduced by Petz \cite{P85, P86} as a quantum generalization of a classical Csisz\'ar's $f$-divergence \cite{C67-2}. It is defined in the context of von Neumann algebras, but we  consider only the Hilbert space setup. Let $\cH$ be a  finite-dimensional  Hilbert space, $\rho$ and $\sigma$ be two states (given by density operators), $K$ be an operator on $\cH$, and $f:(0,\infty)\rightarrow\mathbf{R}$ be an operator convex function. Then the quasi-relative entropy is defined as
$$S_f^K(\rho||\sigma)=\Tr(\rho^{1/2}K^*f(\Delta_{\sigma,\rho})(K\rho^{1.2}))\ , $$
where $\Delta_{\sigma,\rho}$ is a relative modular operator defined by Araki \cite{A76} that acts as a left and right multiplication
$$\Delta_{A,B}(X)=L_AR_B^{-1}(X)=AXB^{-1}\ .$$
The modular operator can be applied straightforward when $\rho$ is invertible. When $\rho$ is not invertible, we take the generalized inverse of $\rho$ and denote it as $\rho^{-1}$ as well. The generalized inverse is defined as follows: if $\sum(\rho)$ denotes the spectrum of $\rho$, and $P_\lambda$ denotes the spectral projection corresponding to the eigenvalue $\lambda$, then the generalized inverse $\rho^{-1}:=\sum_{\lambda\in\sum(\rho)\setminus\{0\}}\lambda^{-1}P_\lambda$.

Note that taking $f(x)=-\log(x)$ and $K=I$ reduces quasi-relative entropy to the Umegaki relative entropy \cite{U62},
$$S(\rho\|\sigma)=\Tr(\rho[\log\rho-\log\sigma])\ . $$

We consider several properties of the relative entropy and their analogue in the case of a quasi-relative entropy.

\textbf{Monotonicity of relative entropy}. The most essential property of the relative entropy is the monotonicity inequality (or data processing inequality). It states that the quasi-relative entropy cannot increase after the states pass through a noisy quantum channel (i.e. a completely-positive, trace-preserving map) $\cN$:
$$S(\rho\|\sigma)\geq S(\cN(\rho)\|\cN(\sigma))\ .$$
This inequality was proved by Lindblad \cite{L75}, building on the work of Lieb and Ruskai \cite{LR73}. Due to a Stinespring factorization \cite{St55}, the monotonicity inequality under any quantum channel, is equivalent to the monotonicity inequality under partial traces, which will be used throughout the paper. Let $\cH=\cH_1\otimes\cH_2$, and $\rho$ and $\sigma$ be two states on $\cH$. We denote the partial trace over the second system as $\rho_1:=\Tr\rho_{12}$, and others similarly. Then the monotonicity inequality states that
$$S(\rho\|\sigma)\geq S(\rho_1\|\sigma_1)\ .$$

 Petz \cite{P86-3, P88} provided a condition on states $\rho$ and $\sigma$, for which  the monotonicity inequality  in saturated. He showed that for a given quantum channel $\cN$, two states $\rho$ and $\sigma$ lead to equality in the monotonicity relation if and only if the channel in noiseless for these states, i.e. the action of the channel can be reversed. The reverse action is implemented by a certain recovery map $\sR$, now called a Petz's recovery map, which can be written explicitly for a given channel. 
 
These results motivated the question of stability of the recovery map: if the decrease of relative entropy after states pass through a quantum channel is small, how well can  these states be recovered? Work on this question started following a 
breakthrough result by Fawzi and Renner in 2015 \cite{FR15}. They proved that if the  strong subadditivity inequality (SSA) (which is equivalent to the monotonicity inequality) is nearly saturated, then  {quantum Markov chain condition}, 
known  to be necessary and sufficient for equality in SSA \cite{H03}, is also nearly satisfied, and they gave a precise quantitative version of this stability result. 

Further refinements of the monotonicity relation occurred later in, for example, \cite{BT16, BHOS15, CL12, JRSWW15, SFR16, STH16, W15, Z16}. Most of 
these results involve some sort of a recovery channel in the lower bound, and even though it often derived from, or related to, the Petz recovery map, it is in no case the Petz map itself.  The quantities provided in those lower bounds are hard to compute.

Recently, Carlen and Vershynina \cite{CV17-1} proved a sharp stability result for the monotonicity inequality in terms of the original Petz recovery map:
$$S(\rho||\sigma) - S(\rho_1||\sigma_1)  \geq \left(\frac{\pi}{8}\right)^{4} \|\Delta_{\sigma,\rho}\|^{-2}
\| \sR_\rho(\sigma_1) -\sigma\|_1^4 \ .
 $$
It is possible to obtain the lower bound in terms of recovery map $\sR_\sigma$, but the constant changes. This bound makes it very easy to see the relation between the saturation of the monotonicity inequality and the perfect recovery of both states. 

\textbf{Monotonicity of quasi-relative entropy}. The monotonicity inequality for a quasi-relative entropy $S_f^K$ also holds, but for operators $K$ of the a certain form. As mentioned before, considering quantum channels is equivalent to considering only a partial trace channel. For operators $K=K\otimes I$ we have
$$S_f^{K\otimes I}(\rho||\sigma)\geq S_f^K(\rho_1\|\sigma_1)\ . $$
This statement was explicitly proved by Sharma \cite{S14}, building on the work of Nielsen and Petz \cite{NP05}. The monotonicity can also be proved for operators $K=K\otimes V$, where $V$ is a unitary.

Sharma has also provided a condition on the equality, showing that the equality in the monotonicity inequality holds if and only if, for all $\beta\in\mathbb{C}$,
$$\Tr(K^*\sigma_1K\rho_1^{1-\beta}-K^*\sigma^\beta K\rho^{1-\beta})=0\ . $$

Building on their previous work, Carlen and Vershynina \cite{CV17-2} also recently showed that for $f$-divergences (i.e. $K=I$), the following sharpening of the monotonicity inequality holds: for a large class of functions $f$ (that define a constant $c$), there is constant $M$ depending only on  the smallest non-zero eigenvalue of $\rho$, $\|\sigma_1^{-1}\|$, $\beta$, and $f$, such that
$$\max\{ \|\sR_\rho(\sigma_1)-\sigma\|_1 \ ,\ \|\sR_\sigma(\rho_1)-\rho\|_1\}   \leq M 
(S_f(\rho||\sigma)-S_f(\rho_1\|\sigma_1))^{\frac14  \frac{1}{ 1+c} }\ .$$
Similar to their previous bound, it is very easy to see here that if the monotonicity is  saturated, Petz's map recovers both states perfectly. The other way is easily derived.

In fact, Carlen and Vershynina showed that even a stronger bound holds: for all $\beta\in(0,1)$, and conditions above, there is a  constant $\alpha(\beta)$ (that also depends on the function $f$) such that
$$\|\sigma_1^{\beta}\rho_1^{-\beta} \rho^{1/2} -\sigma^{\beta}\rho^{1/2-\beta}\|_2\leq M 
(S_f(\rho||\sigma)-S_f(\rho_1\|\sigma_1))^{\alpha(\beta) }\ .
 $$
This means that both: saturation of the monotonicity inequality and Petz's recovery of both states, is equivalent to the following condition: 
$$\text{for all }\beta\in\mathbb{C}: \  \sigma_1^{\beta}\rho_1^{-\beta} =\sigma^{\beta}\rho^{-\beta}\ .$$

We build on Carlen and Vershynina \cite{CV17-2} work to improve the monotonicity result for quasi-relative entropies with an operator $K=K\otimes V$, where $V$ is unitary. We show that with the same conditions as above,
$$\|\sigma_1^{\beta}K\rho_1^{-\beta} \rho^{1/2} -\sigma^{\beta}K\rho^{1/2-\beta}\|_2\leq M 
(S_f^K(\rho||\sigma)-S_f^K(\rho_1\|\sigma_1))^{\alpha(\beta) }\ . $$
In particular, we obtain the bound in terms of the Petz's recovery map
$$ \|\sR_\rho(K_1^*\sigma_1 K_1)-K^*\sigma K\|_1\leq M 
(S_f^K(\rho||\sigma)-S_f^{K_1}(\rho_1\|\sigma_1))^{\frac{1}{4(1+c)} }\ .
$$
Moreover, we show that the equality in the monotonicity inequality holds if and only if, for all $\beta\in\mathbb{C}$,
$$\sigma_1^{\beta}K\rho_1^{-\beta} \rho^{1/2} =\sigma^{\beta}K\rho^{1/2-\beta} \ . $$

\textbf{Strong subadditivity inequality.} For a tri-partite state $\rho_{ABC}$ on a Hilbert space $\cH=\cH_A\otimes\cH_B\otimes\cH_C$, strong subadditivity inequality (SSA) states
$$0\leq S(\rho_{AB})+S(\rho_{BC})-S(\rho_{ABC})-S(\rho_B)\ .$$

This theorem was proved by Lieb and Ruskai \cite{LR73}, using  Lieb's theorem that was proved in \cite{L73}. Note that this inequality is equivalent to the monotonicity inequality: having monotonicity inequality, to obtain the SSA inequality, one takes $\rho=\rho_{ABC}$ and $\sigma=\rho_{AB}\otimes\rho_{C}$ and a trace over the system $\cH_A$. Having SSA, one chooses $\rho_{ABC}$ to be block-diagonal, which implies that the map $\rho_{12}\rightarrow S{\rho_1}-S(\rho_{12})$ is convex. Following Lieb and Ruskai \cite{LR73}, this yields the monotonicity inequality.

We apply previous results in \cite{CV17-1}, to obtain the sharpening of the SSA: for any $\beta\in(0,1)$, there are some constants $N$  (depending on the minimal eigenvalue of $\rho_{ABC}$, $\beta$ and $f$) and $\alpha(\beta)$ (depending on the function $f$) such that
$$N 
\|\rho_B^{\beta}\otimes\rho_{C}^{\beta}\,\rho_{BC}^{-\beta} \rho^{1/2} - \rho_{AB}^{\beta}\otimes\rho_{C}^{\beta}\rho^{1/2-\beta}\|_2^{1/\alpha(\beta)}\leq S(\rho_{AB})+S(\rho_{BC})-S(\rho_{ABC})-S(\rho_B)\ . $$
In particular we show that the equality in SSA holds if and only if Petz's map $\sR$ recovers the reduced state perfectly:
$$\sR_\rho(\rho_B\otimes\rho_{C})=\rho_{AB}\otimes\rho_C\ . $$

\textbf{Operator inequalities of the strong subadditivity type.} In 2012 Kim \cite{K12} proved the following operator version of the strong subadditivity inequality
$$0\leq \Tr_{AB}\rho_{ABC}[\log\rho_{ABC}-\log\rho_{AB}-\log\rho_{BC}+\log\rho_B ]  \ .
 $$
 Note that this inequality leads to the strong subadditivity inequality after taking the trace over system $\cH_C$.
 
 Building on this work, Ruskai \cite{R12} provided a class of operator inequalities: for  an operator monotone decreasing function $f$, 
 \begin{equation}\label{eq:intro-Tr}
 0\leq  \Tr_{AB}\left([f(L_{\sigma_{AB}}R^{-1}_{\rho_{ABC}})-f(L_{\sigma_{B}} R^{-1}_{\rho_{BC}})]\rho_{ABC}\right)\ .
 \end{equation}
 
 We provide the error term for these inequalities of the type: for a large class of operator convex functions $f$, there are some constants $N$  (depending on the minimal eigenvalue of $\rho_{ABC}$, $\beta$ and $f$) and $\alpha(\beta)$ (depending on the function $f$) such that
 $$N [\Tr_{AB}(P_{ABC}(\rho,\sigma)P^*_{ABC}(\rho,\sigma))]^{1/\alpha(\beta)}\leq \Tr_{AB}\left([f(L_{\sigma_{AB}}R^{-1}_{\rho_{ABC}})-f(L_{\sigma_{B}} R^{-1}_{\rho_{BC}})]\rho_{ABC}\right)\ , $$
 where
$$P_{ABC}(\rho,\sigma):=\sigma_{B}^{\beta}\rho_{BC}^{-\beta} \rho_{ABC}^{1/2} -\sigma_{AB}^{\beta}\rho_{ABC}^{1/2-\beta}\ .$$
Note that the left-hand side of this inequality is the operator on $\cH_C$, as is the right-hand side. Therefore, this inequality holds between operators on $\cH_C$.

 Moreover, we show that the equality in (\ref{eq:intro-Tr}) holds if and only if, for all $\beta\in(0,1)$
 $$\sigma_{B}^{\beta}\rho_{BC}^{-\beta} =\sigma_{AB}^{\beta}\rho_{ABC}^{-\beta}\ , $$
 which is in turn equivalent to the recovery condition
 $$ \sR_{\rho_{ABC}}(\sigma_B)=\sigma_{AB}\ .$$

 Additionally, by taking $f$ to be a power function, Ruskai \cite{R12} showed that (\ref{eq:intro-Tr}) leads to
 $$0\leq\frac{1}{p(1-p)}[-\Tr_{AB}\rho_{ABC}^{1-p}\sigma_{AB}^p+\Tr_B\rho_{BC}^{1-p}\sigma_B^p] \ ,$$
 where $p\in(-1,2)$. 
 We apply our results to show the operator strengthening of this inequality:
 $${N}\,[\Tr_{AB}(Q^*_{ABC}(\sigma_{AB},\rho_{ABC})Q_{ABC}(\sigma_{AB},\rho_{ABC})]^{1/\alpha(\beta)}\leq \frac{1}{p(1-p)}\left[-\Tr_{AB}\rho_{ABC}^{1-p}\sigma_{AB}^p+\Tr_B\rho_{BC}^{1-p}\sigma_B^p\right]\ , $$
where
$$Q_{ABC}(\sigma_{AB},\rho_{ABC}):=\rho_{BC}^{\beta}\sigma_{B}^{-\beta} \sigma_{AB}^{1/2} -\rho_{ABC}^{\beta}\sigma_{AB}^{1/2-\beta}\ . $$

\textbf{Joint convexity of the relative entropy.} It was noted by Lindblad \cite{L74} and Ulhmann \cite{U77}, the relative entropy is jointly convex, i.e., if $\rho=\sum_j p_j \rho_j$ and $\sigma=\sum_j p_j \sigma_j$ (with $\sum_j p_j=1$ and $p_j\geq 0$), then
$$  0\leq \sum_j p_j S^K(\rho_j\| \sigma_j)- S^K(\rho\|\sigma)\ .$$

This inequality is also equivalent to the monotonicity inequality: having monotonicity inequality,  one chooses both $\rho$ and $\sigma$ to be block-diagonal, which immediately leads to the joint convexity (see Section~\ref{sec:Joint}). One the other hand, having the joint convexity inequality, Araki and Lieb \cite{AL70, L75} used a purification process to show that on the set of pure states (the extreme points of the set of density matrices) the SSA holds with equality.

\textbf{Joint convexity of the quasi-relative entropy.} If was shown by Petz \cite{OP93, P86-1} or \cite[Theorem 2]{P10} that the quasi relative entropy is jointly convex in $\rho$ and $\sigma$: $$0\leq  \sum_j p_j S_f^K(\rho_j\| \sigma_j)-S_f^K(\rho\|\sigma). $$

Using our strengthening of the monotonicity inequality, we obtain the error term for the joint convexity inequality:  for a large class of operator convex functions $f$, and any $\beta\in(0,1)$, there are some constants $M$ and $\alpha(\beta)$ (see above for the dependence), the following holds
$$\sum_j p_j^{1/2} \|\sigma^{\beta}K\rho^{-\beta} \rho_j^{1/2} - \sigma_j^{\beta}K\rho_j^{1/2-\beta}\|_2\leq M 
\left(\sum_j p_j S_f^K(\rho_j\| \sigma_j) - S_f^K({\rho}\| {\sigma})\right)^{\alpha(\beta) }\ .
 $$
Moreover, we show that the equality in the joint convexity inequality holds if and only if, for all $j$ and all $\beta\in\mathbb{C}$,
$$\sigma^{\beta}K\rho^{-\beta}  = \sigma_j^{\beta}K\rho_j^{-\beta}\ . $$

We apply these results to show the error term in the concavity inequality of the term $\Tr(K^*\sigma^p K\rho^{1-p})$. The concavity of this term was shown by Lieb \cite{L73} for powers of $\rho$ and $\sigma$ that sum up to a number no greater than one.

\textbf{Structure of the paper.} In the next Section~\ref{sec:Op-mon} we review known results for the operator monotone functions, in particular its integral representation. In Section~\ref{sec:Quasi} we introduce the quantum quasi-relative entropy and present few simples, but important facts about it.  In Section~\ref{sec:Mono} we review the previous strengthening of the monotonicity inequality for relative entropy and $f$-divergences, and present the error term for the quasi-relative entropy, with the condition for the equality. In Section~\ref{sec:Joint} we apply the strengthening of the monotonicity inequality established in the previous section to the joint convexity inequality and provide the condition for the equality as well. In Section~\ref{sec:ssa} we provide the error term for a class of operator inequalities established by Ruskai \cite{R12} with the condition for equality. In Section~\ref{sec:log} we apply all previous results to obtain new inequalities for the relative entropy by taking the function $f(x)=-\log(x)$ and $K=I$. In particular, we obtain the error terms for the joint convexity inequality, strong subadditivity inequality, and the operator strong subadditivity inequalities. In Section~\ref{sec:WYD} we apply previous results to the power function. The quasi-relative entropy for the power functions gives a term $\Tr(K^*\sigma^p K\rho^{1-p})$, concavity of which leads to the concavity of the Wigner-Yanase-Dyson $p$-skew information. In particular, we provide a Pinsker inequality for such a term, and the error terms for the joint concavity and the operator version of WYD inequalities. At the end, we apply these results to show the error term for the operator Cauchy-Schwartz inequality.

%%%%%%%%%%%%%
%%%%%%%
\section{Operator monotone functions}\label{sec:Op-mon}

\begin{definition}\label{def:op-mon}
A function $f: (a,b) \to \R$ is {\em operator monotone} if for any pair of  self-adjoint operators 
 $A$ and $B$ on some 
 Hilbert space that have spectrum in $(a,b)$, the operator $$f(A) -f(B)\geq0$$ is positive semidefinite whenever $A - B\geq 0$ is positive semidefinite. We say that $f$ is {\it operator monotone decreasing} on $(a,b)$ in case $-f$ is 
 operator monotone.
\end{definition}

\begin{definition}\label{def:convex}
A function $f$ is {\em operator concave} on the positive operators, when  for all positive semidefinite 
operators $A$ and $B$, and all $\lambda$ in $(0,1)$, $$ f((1-\lambda)A + 
\lambda B))-(1-\lambda)f(A) - \lambda f(B)\geq 0$$
is positive semidefinite. A function $f$ is {\em operator convex} when $-f$ is operator concave.
\end{definition}

\begin{theorem}[Bhatia '97 ]\cite[Theorem V.2.5]{B97}
Every continuous function $f$ mapping $[0,\infty)$ into itself is operator monotone if and only if it is operator concave. 
\end{theorem}

\begin{definition}\label{def:pick}
 A {\em Pick function} is a function $f$ that is analytic on the upper half plane and has a positive imaginary part. The set of Pick functions on $(a,b)$ is denoted as $\mathcal{P}_{(a,b)}.$
\end{definition}

\begin{theorem}[L\"{o}wner '34] \cite[Theorem V.4.7]{B97} A function $f$ on $(a,b)$ is operator monotone if and only if $f$ is a restriction of a pick function $f\in\mathcal{P}_{(a,b)}$ to $(a,b)$.
\end{theorem}

According to \cite[Chapter II, Theorem I]{Dono} every operator monotone decreasing function $f$ (i.e. $-f\in\mathcal{P}_{(0,\infty)}$) has a canonical integral representation 
\begin{equation}\label{low}
f(x) = a x{+} b {-}\int_{0}^\infty \left(  \frac{1}{t +x }-\frac{t}{t^2+1}   \right){\rm d}\mu_f(t)\ ,
\end{equation}
where  $a\leq 0$, $b\in\bR$ and $\mu$ is a positive measure on $(0,\infty)$ such that 
${\displaystyle \int_{0}^\infty  \frac{1}{t^2+1}{\rm d}\mu_f(t)<\infty}$.
Conversely, every such function is operator monotone decreasing.

 The following formulas \cite[Chapter II, p. 24]{Dono}  determine $a$, $b$ and $\mu$ corresponding to $f$.
\begin{equation}\label{alphabeta}
a = \lim_{y\uparrow\infty}\frac{f(iy)}{iy} \qquad{\rm and}\qquad b = \Re(f(i))\ .
\end{equation}

Define the monotone increasing function $\mu(x) := \frac12 \mu(\{x\}) + \mu((-\infty, x))$, then
according to \cite[Chapter II, Lemma 2] {Dono}  we have that
 \begin{equation}\label{muform}
 \mu(x_1) - \mu(x_0) = \lim_{y\downarrow 0} \frac{1}{\pi} \int_{x_0}^{x_1} \Im f(-x+iy){\rm d} x\ .
 \end{equation}

\begin{definition}\label{def:reg} A operator monotone function $f$ is {\em regular} in case the measure 
$\mu$ in the canonical integral representation of $f$ is absolutely continuous with respect to 
Lebesgue measure, and for each $S,T > 0$, there is a finite constant $C^f_{S,T}$ such that
\begin{equation}\label{regdef}
{\rm d} t \leq C^f_{S,T} {\rm d}\mu(t)
\end{equation}
on the interval $[1/S,T]$.  An operator monotone decreasing function is {\em regular} if and only if $-f$ is regular. 
\end{definition}

\begin{example}\label{ex-log} Let $f(x)=-\log(x)$. It is operator monotone decreasing.
Then 
$$b=\Re(\log(i)) = 0\ ,$$ and 
$$a=\lim_{y\uparrow\infty}\log(iy)/(iy) = \lim_{y\uparrow\infty}(\log y + i\pi/2)/(iy) =0\ . $$ It is clear from \eqref{muform} that
$${\rm d}\mu(x) = \frac{1}{\pi}\lim_{y\downarrow 0}\Im\log(-x + iy){\rm d}x = {\rm d}x\ .$$
Then the  integral representation (\ref{low}) gives the following formula for the logarithmic function
\begin{equation}\label{eq:logex}
-\log x=\int_{0}^\infty \left( \frac{1}{t +x } - \frac{t}{t^2+1}  \right){\rm d}t \ .
\end{equation}
\end{example}

\begin{example}\label{ex-power}
Let $f(x)=x^p$, where ${p}\in\bR$. Then by \cite[Theorem V.2.10]{B97} the function $f$ is 
\begin{enumerate}
\item operator monotone if and only if $p\in[0,1]$;
\item operator convex  if and only if $p\in[-1,0]\cup[1,2]$;
\end{enumerate}
Consider the power function $f(x)=-x^p$ for $p\in(0,1)$. It is operator monotone decreasing. Then 
$$a = \lim_{y\uparrow\infty}f(iy)/(iy) = 0\ , \ \ \text{and } \ b = {\cos}(p \pi/2)\ .$$

For $x>0$, $\lim_{y\downarrow 0} \Im f(-x + iy) = x^p\sin(p \pi)$ so that
$${\rm d}\mu(x) = \pi^{-1}\sin(p \pi) x^p{\rm d}x\ .$$ This yields the representation
\begin{equation}\label{eq:powex}
-x^p =  - {\cos}(p \pi/2)  + \frac{\sin(p \pi)}{\pi} \int_{0}^\infty t^p\left(  
\frac{1}{t +x } -\frac{t}{t^2+1}  \right){\rm d}t \ .
\end{equation}
\end{example}

%%%%%%%%%%%%%%%%%%%%
%%%%%%%%%%%%%%%%%%%
\section{Quantum quasi-relative entropy}\label{sec:Quasi}

The notion of quantum quasi-relative entropy was introduced by Petz \cite{P85, P86-1} or \cite[Chapter 7]{OP93}. Let $\cH$ be a finite-dimensional Hilbert space, and $\rho$ and $\sigma$ be states on $\cH$ (i.e. trace-one, positive semi-definite operators). Note that notions and results in this paper can be formulated for von Neumann algebras, as it was done in multiple references. 
\begin{definition}
For an operator monotonically decreasing function $f$ (which implies that $f$ is also operator convex), such that $f(1)=0$,  states $\rho$ and $\sigma$, and a bounded operator $K$ on $\cH$, {\it the quasi-relative entropy} is defined as 
$$S_f^K(\rho|| \sigma)=\Tr(\sqrt{\rho}K^*f(\Delta_{\sigma,\rho})K\sqrt{\rho}),$$
where the modular operator, introduced by Araki \cite{A76}, $$\Delta_{A,B}(X)=L_AR_B^{-1}(X)=AXB^{-1}$$ is a product of left and right multiplication operators, $L_A(X)=AX$ and $R_B(X)=XB$.
\end{definition}

There is a straightforward way to calculate the quasi-relative entropy from the spectral decomposition of states \cite{V16}. Let $\rho$ and $\sigma$ have the following spectral decomposition
$$\rho=\sum_j\lambda_j\ket{\psi_j}\bra{\psi_j}, \ \ \sigma=\sum_k\mu_k\ket{\phi_k}\bra{\phi_k}.$$
Then the quasi-relative entropy is calculated as follows  
\begin{equation}\label{eq:formula}
S_f^K(\rho||\sigma)=\sum_{j,k}\lambda_j f\left(\frac{\mu_k}{\lambda_j}\right)|\bra{\phi_k}K\ket{\psi_j}|^2. 
\end{equation}

\begin{example}\label{ex:quasi}
\begin{enumerate}
\item For $K=I$, the quasi-relative entropy
$$S_f(\rho|| \sigma)=\Tr({\rho}^{1/2}f(\Delta_{\sigma,\rho}){\rho}^{1/2}),$$
is sometimes referred to as an $f$-divergence.
\item\label{ex:J}
For $-1\leq p<1$ define a function $$f_p(x):=\left\{\begin{matrix}
\frac{1}{p(1-p)}(1-x^p) & p\neq 0\\
-\log x & p=0
\end{matrix}\right..$$
Note that from Example~\ref{ex-power} the function is convex. The quasi-relative entropy for this function
$$S_{f_{1-p}}^{K^*}(\rho||\sigma)=J_p(K, \rho, \sigma):=\Tr (\sqrt{\sigma}K^*g_p(\Delta_{\rho,\sigma})K\sqrt{\sigma})$$
is the function $J_p$ defined by Jencova and Ruskai in \cite{JR10}, here for $0<p\leq 2$ the function $g_p$ is defined as
 $$g_p(x):=xf_{1-p}(x^{-1})=\left\{\begin{matrix}
\frac{1}{p(1-p)}(x-x^p) & p\neq 1\\
x\log x & p=1
\end{matrix}\right..$$ 
\item In the above example, if $p=0$ (i.e. $f(x)=-\log x$) and $K=I$, we obtain the Umegaki relative entropy \cite{U62}
$$S_{f_0}^I(\rho||\sigma)={J}_1(I,\rho,\sigma)=S(\rho||\sigma)=\Tr(\rho\log \rho- \rho\log \sigma). $$
\item In example \ref{ex:J}, the quasi-relative entropy  can be calculated to be
$${S_{f_p}^K(\rho|| \sigma)=\frac{1}{p(1-p)}\Tr(K^*\rho K- K^*\sigma^{p}K\rho^{1-p})}\ .$$
This expression has the term $\Tr(K^*\sigma^{p}K\rho^{1-p})$, concavity of which was proved by Lieb in \cite{L73} with more general powers.
\item\label{ex:WYD} In the above example, taking $\sigma=\rho$ and $K^*=K$, results in 
$$S_{f_{p}}^K(\rho|| \rho)=-\frac{1}{2p(1-p)}\Tr[K, \rho^p][K, \rho^{1-p}]\geq 0 \ .$$
Up to a constant, this is the Wigner-Yanase-Dyson $p$-skew information \cite{WY63, WY64}, for $p\in(0,1)$.

\end{enumerate}
\end{example}

Now we state a few simple properties of quasi-relative entropy, some of which have been noted before.
\begin{proposition}
The quasi-relative entropy scales as follows: for any constant $c$,
$$S_f^K(c\rho\|c\sigma)=cS_f^K(\rho\|\sigma) \ .$$
\end{proposition}
\begin{proof}
This following directly from the formula (\ref{eq:formula}).
\end{proof}

\begin{proposition}\label{prop:U}
For a unitary $U$, and states $\rho$ and $\sigma$
$$S_f^U(\rho||\sigma)=S_f^I(U\rho U^*||\sigma)=S_f^I(\rho|| U^*\sigma U). $$
\end{proposition}
\begin{proof}
If $\rho$ and $\sigma$ have the following spectral decomposition
$$\rho=\sum_j\lambda_j\ket{\psi_j}\bra{\psi_j}, \ \ \sigma=\sum_j\beta_j\ket{\phi_j}\bra{\phi_j},$$
then by \cite{V16}
$$f(\Delta_{\sigma,\rho})=\sum_{j,k}f\left( \frac{\beta_j}{\lambda_k}\right)\Delta_{\ket{\phi_j}\bra{\phi_j},\ket{\psi_k}\bra{\psi_k}}.$$
Therefore, the quasi-relative entropy can be written as
\begin{align}
S_f^U(\rho||\sigma)&=\sum_{j,k}f\left( \frac{\beta_j}{\lambda_k}\right)\Tr(\sqrt{\rho}\,U^*\ket{\phi_j}\bra{\phi_j}U\sqrt{\rho}\ket{\psi_k}\bra{\psi_k}).
\end{align}
On the other hand, using the spectral decomposition of $\rho$ and the fact that $U$ is unitary, we obtain
\begin{align}
S_f^I(U\rho U^*||\sigma)&=\Tr(U\sqrt{\rho}\,U^*f(\Delta_{\sigma,U\rho U^*})(U\sqrt{\rho}\,U^*))\\
&=\sum_{j,k}f\left( \frac{\beta_j}{\lambda_k}\right)\Tr(U\sqrt{\rho}\,U^*\ket{\phi_j}\bra{\phi_j}U\sqrt{\rho}\,U^*U\ket{\psi_k}\bra{\psi_k}U^*)\\
&=\sum_{j,k}f\left( \frac{\beta_j}{\lambda_k}\right)\Tr(\sqrt{\rho}\,U^*\ket{\phi_j}\bra{\phi_j}U\sqrt{\rho}\ket{\psi_k}\bra{\psi_k})\\&=S_f^U(\rho||\sigma).
\end{align}
Similarly, $S_f^U(\rho||\sigma)=S_f^I(\rho|| U^*\sigma U). $
\end{proof}

\begin{proposition}
For a unitary $U$, and states $\rho$ and $\sigma$ the quasi-relative entropy is non-negative
$$S_f^U(\rho||\sigma)\geq 0.$$
\end{proposition}
\begin{proof}
Form the previous proposition it is evident that it is enough to consider $U=I$. Fix a basis $\{\ket{j}\bra{j}\}_j$ of the Hilbert space $\cH$ the states act on, and consider the map
$$\Phi(\omega)=\sum_j\bra{j}\omega\ket{j}\ket{j}\bra{j}. $$
It is a Schwarz map, so from the monotonicity of the quasi-relative entropy (\ref{eq:original-mono}), which we will discuss in Section \ref{sec:Mono-review}, we have
$$S_f^I(\rho|| \sigma)\geq S_f^I(\Phi(\rho)|| \Phi(\sigma)). $$
The last expression is the classical quasi-relative entropy, which is non-negative.
\end{proof}

\begin{proposition}
For a unitary $U$, and states $\rho$ and $\sigma$
$$S_f^U(\rho||\sigma)= 0 \ \ \text{if and only if}\ \ \rho=U^*\sigma U.$$
\end{proposition}
\begin{proof}
According to Proposition~\ref{prop:U}, it is enough to consider $U=I$. 

Then, one way: if $\rho=\sigma$, then $S_f(\rho,\sigma)=0$ by definition, since $f(1)=0$.

The other way: from the proof of last proposition, 
$$0=S_f^I(\rho|| \sigma)\geq S_f^I(\Phi(\rho)|| \Phi(\sigma)). $$
Therefore, for any basis $\{\ket{j}\}_j$, $\Phi(\rho)=\Phi(\sigma)$, which implies that $\rho=\sigma$.
\end{proof}

A famous bound relating quantum relative entropy and trace distance between two quantum states, called Pinsker inequality. The similar inequality holds for the quasi-relative entropy as well, as was shown by Hiai and Mosonyi \cite{HM16} for $U=I$.
\begin{proposition} \label{prop:pinsker}
For an operator convex function $f$ on $(0,\infty)$ with $f(1)=0$, a unitary $U$, and states $\rho$ and $\sigma$, the following holds
$$\frac{f''(1)}{2}\|\rho-U^*\sigma U\|_1^2\leq S_f^U(\rho\|\sigma)\ . $$
\end{proposition}

It's worthwhile to point out that any inequality that holds in the classical case for probability distributions, also holds in the quantum case between density operators.  Let us show how to obtain quantum inequality from the classical one. This relation can be found in many quantum information books, e.g. \cite{W18}, for relative entropies, i.e. when $f(x)=-\log(x)$. For a general function $f$ the proof is similar.

\begin{lemma}\label{lemma:eq}
If the following inequality holds for all probability distributions $p$ and $q$:
$${S_f^{cl}(p\|q)\geq F(\|p-q\|_1)}, $$
for some {function $F$}, then for all states $\rho$ and $\sigma$ the following holds as well
$${S_f(\rho\|\sigma)\geq F(\|\rho-\sigma\|_1)}\ . $$
\end{lemma}
\begin{proof}
For an operator $\rho-\sigma$ consider its Jordan-Hahn decomposition
$$\rho-\sigma=P-Q,$$
where $P,Q>0$. Define a projector onto the image of $P$  as
$\Pi=\Pi_{im(P)}.$ 
Then, a measurement with operators $\{\Pi, I-\Pi\}$ is a projective measurement. It holds that the trace distance between states $\rho$ and $\sigma$ equals to the $L^1$ distance between probability distributions
$$\|\rho-\sigma\|_1=\Tr(P)+\Tr(Q)=2\Tr(P)=2|\Tr(\Pi(\rho-\sigma))|=\|p-q\|_1,$$
where $p=\{\Tr(\Pi\rho), 1-\Tr(\Pi\rho)\}$, $q=\{\Tr(\Pi\sigma), 1-\Tr(\Pi\sigma)\}$ are probability ensembles. 

For above $\Pi$, define a quantum-to-classical channel as follows
$$\Phi(\omega)=\Tr(\Pi\omega)\ket{0}\bra{0}+(1-\Tr(\Pi\omega))\ket{1}\bra{1}. $$
Then from the monotonicity inequality for the quasi-relative entropy, and for above $p$ and $q$, the quantum quasi-relative entropy is {lower bounded by}  the classical quasi-relative entropy
$$S_f(\rho\|\sigma){\geq S_f(\Phi(\rho)\|\rho(\sigma))}=\sum_jp_jf(p_j^{-1}q_j)=S_f^{cl}(p\|q).$$
Therefore, any inequality that holds between classical $f$-divergence and the $L^1$ distance, also holds in the quantum case between quantum $f$-divergence and trace distance between two states.
\end{proof}

%%%%%%%%%%%%%%%%%%
%%%%%%%%%%%%%%%%
\section{Monotonicity inequality}\label{sec:Mono}

In this section we assume a bipartite Hilbert space $\cH=\cH_1\otimes\cH_2$. Two states $\rho$ and $\sigma$ act on this Hilbert space. Denote $\rho_1:=\Tr_2\rho$ and similarly, $\sigma_1=\Tr_2\sigma$.
 Here we build on results in \cite{CV17-2} to provide a strengthening of the monotonicity inequality for quasi-relative entropies for a large class of operators $K$.

\subsection{Review}\label{sec:Mono-review}

For a regular relative entropy, i.e. $S(\rho\|\sigma)=\Tr[\rho(\log\rho-\log\sigma)]$, Lindblad proved \cite{L74} the following monotonicity inequality

\begin{theorem}[Monotonicity of relative entropy]
 \begin{equation}\label{eq:original-mono-reg}
 S(\rho||\sigma) -S(\rho_1 || \sigma_1)\geq 0.
 \end{equation}
\end{theorem}
Lindblad showed that the monotonicity inequality \eqref{eq:original-mono-reg} is equivalent to the joint convexity of the relative entropy $(\rho,\sigma) \mapsto
S(\rho||\sigma)$, see Proposition~\ref{prop:convex}, and this in turn is an immediate consequence of Lieb's Concavity Theorem \cite{L73}, which was proved by Lieb and Ruskai \cite{LR73}, who showed it to be equivalent to the Strong Subadditivity (SSA) of the von Neumann entropy, see Theorem~\ref{thm:ssa-reg}.

Note that the monotonicity inequality also holds for CPTP (completely-positive trace-preserving) maps, not just a partial trace. This fact was proved by Lindblad \cite{L75} by using Stinesping's Dialtion Theorem \cite{St55} that relates general CPTP maps to a partial trace. Therefore, we will focus only on partial traces, not general CPTP maps.

Petz has proved \cite{P86-3, P88} that for a relative entropy there is an equality in the monotonicity inequality (\ref{eq:original-mono-reg}) if and only if both states $\rho$ and $\sigma$ can be  recovered perfectly. The recovery map $\sR_\rho$ is known as {\it Petz recovery map} and is defined as $\sR_\rho:\cB(\cH_1)\rightarrow \cB(\cH=\cH_1\otimes\cH_2)$ 
\begin{equation}\label{eq:Petz}
\sR_\rho(\gamma) = \rho^{1/2} \rho_1^{-1/2}\gamma\rho_1^{-1/2}\rho^{1/2}\ .
\end{equation}

There has been several results that provide a lower bound in (\ref{eq:original-mono}) other than zero \cite{FR15, JRSWW15, W15, Z16}, but the lower bounds provided there involve quantities that are hard to compute, e.g. rotated and twirled Petz recovery maps. For another fidelity type bound not explicitly involving the recovery map, see \cite[Theorem 2.2]{CL14}. 

In 2017, Carlen and Vershynina \cite{CV17-1} provided the following sharpening of the monotonicity inequality
\begin{equation}\label{eq:CV17-1}
S(\rho||\sigma) - S(\rho_1||\sigma_1)  \geq \left(\frac{\pi}{4}\right)^{4} \|\Delta_{\sigma,\rho}\|^{-2}
\|\sigma_1^{1/2}\rho_1^{-1/2} \rho^{1/2} - \sigma^{1/2}\|_2^4\  ,
\end{equation}
with $\|\cdot\|_2$ denoting the Hilbert-Schmidt norm
\begin{equation}\label{eq:HS-norm}
\|A\|_2:=\Tr(A^*A).
\end{equation}
As a corollary they provided a bound that explicitly involves Petz's recovery map:
\begin{equation}\label{eq:CV17-1-cor}
S(\rho||\sigma) - S(\rho_1||\sigma_1)  \geq \left(\frac{\pi}{8}\right)^{4} \|\Delta_{\sigma,\rho}\|^{-2}
\| \sR_\rho(\sigma_1) -\sigma\|_1^4\  ,
\end{equation}
and 
\begin{equation}
S(\rho||\sigma) - S(\rho_1||\sigma_1)  \geq \left(\frac{\pi}{8}\right)^{4} \|\Delta_{\sigma,\rho}\|^{-2}
\|\rho_{1}\|^{-2}\|\sigma_1^{-1}\|^{-2}
\| \sR_\sigma(\rho_1) -\rho\|_1^4\  ,
\end{equation}
with $\|\cdot\|_1$ denoting the trace norm
\begin{equation}\label{eq:HS-norm}
\|A\|_1:=\Tr |A|.
\end{equation}
This is the first time when the original Petz recovery map $\sR$ appeared in the lower bound of the monotonicity or equivalent inequality.

For quasi-relative entropies $S_f^K$ the monotonicity inequality holds for operators $K$ such that $K=K_1\otimes V_2$ for a unitary $V$. This was explicitly proved by Sharma \cite{S14}, building on the work of Nielsen and Petz \cite{NP05}.
 
 \begin{theorem}[Monotonicity of quasi-relative entropy, Sharma '14]\label{thm:original-mono}  For every operator convex function $f$ and operator $K=K_1\otimes V$, where $V$ is a unitary, we have
 \begin{equation}\label{eq:original-mono}
 S_f^K(\rho||\sigma) -S_f^{K_1}(\rho_1 || \sigma_1)\geq 0.
 \end{equation}
 \end{theorem}
In Section \ref{sec:proof-mono} we will recall the proof of this theorem, as we rely on it later.

In \cite{CV17-2} Carlen and Vershynina, generalized their previous result (\ref{eq:CV17-1}) for $f$-divergences for regular functions $f$, see Definition~\ref{def:reg}. 
In results below, assume that $T>0$, $\beta\in(0,1)$, and \\
\smallskip
\noindent{\it (1)} for $\beta\leq1/2$, define $T_L(\beta):=T$ and $T_R(\beta):=T^{\beta/(1-\beta)}$;

\smallskip
\noindent{\it (2)} for $\beta\geq1/2$,  define $T_L(\beta):=T^{{(1-\beta)}/{\beta}}$ and $T_R(\beta):=T$.\\
Moreover, define $C_{T, \beta}^f$ to be the least positive constant such that  ${\rm d}t\leq C^f_{T, \beta}\,{\rm d}\mu_f(t)$ for 
$t\in[T_L(\beta)^{-1}, T_R(\beta)]$, noting that $C_{T, \beta}^f>0$ since $f$ is regular.  

In \cite{CV17-2} it is proved that under the above conditions for a regular operator monotone decreasing function $f$ the following holds \\
\begin{eqnarray}\label{eq:CV-mono}
&&\frac{\pi}{\sin{\beta\pi}}\|\sigma_1^{\beta}\rho_1^{-\beta} \rho^{1/2} -\sigma^{\beta}\rho^{1/2-\beta}\|_2\\
&\leq& 2\left(\frac{1}{\beta}+\frac{\|\Delta_{\sigma,\rho}\|}{1-\beta}\right)\frac{1}{T^{\alpha_1(\beta)}}+T^{\alpha_2(\beta)}\,\left(C_{T, \beta}^f\right)^{1/2}\, (S_f(\rho||\sigma)-S_f(\rho_1\|\sigma_1))^{1/2}\ .
\end{eqnarray}
where 
$$\alpha_1(\beta)=\left\{\begin{matrix}
\beta &\text{when }\beta\leq1/2\\
1-\beta &\text{when }\beta\geq1/2.
\end{matrix}\right. \ \ \text{and} \ \ 
\alpha_2(\beta)=\left\{\begin{matrix}
\frac{1-\beta}{2}+\frac{\beta^2}{2(1-\beta)} &\text{when }\beta\leq1/2\\
\beta&\text{when }\beta\geq1/2.
\end{matrix}\right. $$

Optimizing in $T$ for functions, for which $C_{T,\beta}^f$ scales as a power of $T^c$, it was straightforward to prove \cite[Corollary 4.4]{CV17-2} that there is a constant  $M$, depending only on  the smallest non-zero eigenvalue of $\rho$, $\beta$, $C$ and $c$, such that\\
\begin{eqnarray}\label{eq:cor}
\|\sigma_1^{\beta}\rho_1^{-\beta} \rho^{1/2} -\sigma^{\beta}\rho^{1/2-\beta}\|_2\leq M 
(S_f(\rho||\sigma)-S_f(\rho_1\|\sigma_1))^{\alpha(\beta) }\ ,
\end{eqnarray}
where 
$$\alpha(\beta)=\left\{\begin{matrix}
\frac{\beta(1-\beta)}{1+2c(1-\beta)} &\text{when }\beta\leq1/2\\
\frac12  \frac{1-\beta}{ 1+c}&\text{when }\beta\geq1/2.
\end{matrix}\right. $$
Moreover, taking $\beta=1/2$, for a constant $M$ depending only on  the smallest non-zero eigenvalue of $\rho$, $\|\sigma_1^{-1}\|$, $\beta$, $C$ and $c$, we have
\begin{equation}
\max\{ \|\sR_\rho(\sigma_1)-\sigma\|_1 \ ,\ \|\sR_\sigma(\rho_1)-\rho\|_1\}   \leq M 
(S_f(\rho||\sigma)-S_f(\rho_1\|\sigma_1))^{\frac14  \frac{1}{ 1+c} }\ .
\end{equation}
From this expression it is evident that the monotonicity inequality is saturated for a broad class of operator monotone decreasing functions $f$, if and only if the Petz recovery map recovers both states $\rho$ and $\sigma$ perfectly well.

%%%%%%%%%%%%%
\subsection{Monotonicity inequality}\label{sec:mono-new}

We generalize the  result in \cite{CV17-2} to include quasi-relative entropies with a large class of operators $K$. 

\begin{definition}
For an operator monotone decreasing function $f$ assume that $T>0$, $\beta\in(0,1)$, and \\
\smallskip
\noindent{\it (1)} for $\beta\leq1/2$, define $T_L(\beta):=T$ and $T_R(\beta):=T^{\beta/(1-\beta)}$;

\smallskip
\noindent{\it (2)} for $\beta\geq1/2$,  define $T_L(\beta):=T^{{(1-\beta)}/{\beta}}$ and $T_R(\beta):=T$.\\
Moreover, define $C_{T, \beta}^f$ to be the least positive constant such that  ${\rm d}t\leq C^f_{T, \beta}\,{\rm d}\mu_f(t)$ for  $t\in[T_L(\beta)^{-1}, T_R(\beta)]$, noting that $C_{T, \beta}^f>0$ since $f$ is regular.  We will call such functions {\it $C_{T, \beta}^f$-regular.}
\end{definition}

\begin{theorem}\label{thm:mono}
Let $\cH=\cH_1\otimes\cH_2$, and let an operator $K$ be such that $K=K_1\otimes V_2$, where $V$ is a unitary operator. Let $f$ be a $C_{T, \beta}^f$-regular function. Then for any states $\rho, \sigma$ on $\cH$  and any $\beta\in(0,1)$
\begin{eqnarray}\label{eq:thm-mono}
&&\frac{\pi}{\sin{\beta\pi}}\|\sigma_1^{\beta}K\rho_1^{-\beta} \rho^{1/2} -\sigma^{\beta}K\rho^{1/2-\beta}\|_2\\
&\leq& 2\left(\frac{ {\|K\|}}{\beta}+\frac{\|\Delta_{\sigma,\rho}\|}{1-\beta}\right)\frac{1}{T^{\alpha_1(\beta)}}+T^{\alpha_2(\beta)}\,\left(C_{T, \beta}^f\right)^{1/2}\, (S_f^K(\rho||\sigma)-S_f^{K_1}(\rho_1\|\sigma_1))^{1/2}\ .
\end{eqnarray}
where 
$$\alpha_1(\beta)=\left\{\begin{matrix}
\beta &\text{when }\beta\leq1/2\\
1-\beta &\text{when }\beta\geq1/2.
\end{matrix}\right. \ \ \text{and} \ \ 
\alpha_2(\beta)=\left\{\begin{matrix}
\frac{1-\beta}{2}+\frac{\beta^2}{2(1-\beta)} &\text{when }\beta\leq1/2\\
\beta&\text{when }\beta\geq1/2.
\end{matrix}\right. $$
\end{theorem}
The proof of this theorem is given in Section \ref{sec:proof-mono}. Note that if $K=I$ we arrive precisely at the statement in \cite{CV17-2}, i.e. (\ref{eq:CV-mono}).
Similarly, in order to optimize in $T$ one would need more information about the function $f$. For instance, when $C_{T}^f$ grows like a power of $T$, 
the optimization is very straightforward. Using \cite[Lemma 4.3]{CV17-2}, we obtain the following corollary:

\begin{corollary}\label{optimcl} Let $\beta\in(0,1)$ and $f$ be a $C_{T, \beta}^f$-regular function. Let $K=K_1\otimes V_2$ with a unitary $V$. Suppose that $C^f_{T, \beta} \leq C\, T^{2c}$ for some $c, C>0$. Then there is an explicitly computable  constant 
$M$ depending only on the smallest non-zero eigenvalue of $\rho$, $\beta$,  {$\|K\|$}, $C$ and $c$, such that, 
\begin{eqnarray}\label{eq:cor1}
\|\sigma_1^{\beta}K\rho_1^{-\beta} \rho^{1/2} -\sigma^{\beta}K\rho^{1/2-\beta}\|_2\leq M 
(S_f^K(\rho||\sigma)-S_f^{ {K_1}}(\rho_1\|\sigma_1))^{\alpha(\beta) }\ ,
\end{eqnarray}
where 
$$\alpha(\beta)=\left\{\begin{matrix}
\frac{\beta(1-\beta)}{1+2c(1-\beta)} &\text{when }\beta\leq1/2\\
\frac12  \frac{1-\beta}{ 1+c}&\text{when }\beta\geq1/2.
\end{matrix}\right. $$

In particular, for $\beta=1/2,$
\begin{equation}
\|\sigma_1^{1/2}K\rho_1^{-1/2} \rho^{1/2} - \sigma^{1/2}K\|_2 \leq M 
(S_f^K(\rho||\sigma)-S_f^{K_1}(\rho_1\|\sigma_1))^{\frac{1}{4(1+c)} }\ .
\end{equation}
\end{corollary}

Moreover, for $\beta=1/2$ we can relate the expression $\|\sigma_1^{1/2}K\rho_1^{-1/2} \rho^{1/2} - \sigma^{1/2}K\|_2 $ with the one involving Petz recovery map:

\begin{corollary}\label{cor:mono1}
Let $f$ be a $C_{T, 1/2}^f$-regular function. Suppose that $C^f_{T, 1/2} \leq C\, T^{2c}$ for some $c, C>0$.  Then there is an explicitly computable  constant 
$M$ depending only on the smallest non-zero eigenvalue of $\rho$,  {$\|K\|$}, $C$ and $c$, such that, \begin{equation}\label{eq:mono}
\|\sR_\rho(K_1^*\sigma_1 K_1)-K^*\sigma K\|_1\leq M 
(S_f^K(\rho||\sigma)-S_f^{K_1}(\rho_1\|\sigma_1))^{\frac{1}{4(1+c)} }\ .
\end{equation}
\end{corollary}
\begin{proof}
In Lemma 2.2 in \cite{CV17-1} take $$X=\sigma_1^{1/2}K\rho_1^{-1/2}\rho^{1/2}$$ and $$Y=\sigma^{1/2}K.$$ 
Since $V$ is unitary, $K^*\sigma_1 K=K_1\sigma_1 K_1\otimes I$. Therefore, we have that
$$\|\sigma_1^{1/2}K\rho_1^{-1/2} \rho^{1/2} - \sigma^{1/2}K\|_2\geq \frac{1}{2}\|\sR_\rho(K_1^*\sigma_1 K_1)-K^*\sigma K\|_1,$$
where $\sR_\rho$ is a Petz recovery map.
\end{proof}

Note that if $K$ is invertible, we may interchange the roles of $\rho$ and $\sigma$:

\begin{corollary} Let an invertible operator $K$ be such that $K=K_1\otimes V$, where $V$ is a unitary operator. Let $f$ be a $C_{T, 1/2}^f$-regular function. Suppose that $C^f_{T, 1/2} \leq C\, T^{2c}$ for some $c, C>0$.   Then there is an explicitly computable  constant 
$M$ depending only on the smallest non-zero eigenvalue of $\rho$,  {$\|K\|$}, $C$ and $c$, such that, \begin{align}
&\frac{1}{2}\|\rho_1\|^{-1/2}\|K^{-1}\|^{-1} \|K^*\sR_\sigma\left(\left(K_1^{-1}\right)^{-1}\rho_1 K_1^{-1}\right)K-\rho \|_1\leq M 
(S_f^K(\rho||\sigma)-S_f^{K_1}(\rho_1\|\sigma_1))^{\frac{1}{4(1+c)} }\ .
\end{align}
\end{corollary}
\begin{proof}
Recalling $L_A$ being a left multiplication operation, 
$$L_{\rho_1^{1/2}}L_{K_{-1}} L_{\sigma_1^{-1/2}}(\sigma_1^{1/2}K\rho_1^{-1/2} \rho^{1/2} - \sigma^{1/2}K) = 
 \rho^{1/2} -  \rho_1^{1/2}K^{-1}  \sigma_1^{-1/2}\sigma^{1/2}K \ ,$$
and hence
\begin{equation}\label{compA}
\| \rho^{1/2} -  \rho_1^{1/2}  K^{-1}\sigma_1^{-1/2}\sigma^{1/2}K\|_2 \leq \|L_{\rho_1^{1/2}}\|\|L_{K^{-1}}\| \| L_{\sigma_1^{-1/2}}\| \|\sigma_1^{1/2}K\rho_1^{-1/2} \rho^{1/2} - \sigma^{1/2}K\|_2\ .
\end{equation}
Since $\|L_{\rho_1^{1/2}}\| = \|\rho_{1}\|^{1/2}$, $\| L_{\sigma_1^{-1/2}}\|  = \|\sigma_1^{-1}\|^{1/2}$, and $\|L_{K^{-1}}\|=\|K^{-1}\|$
we may combine \eqref{compA} with Theorem \ref{thm:mono} to obtain

\begin{align}
&\|\rho_1\|^{-1/2}\|K^{-1}\|^{-1} \|\sigma_1^{-1}\|^{1/2} \|\rho^{1/2} -  \rho_1^{1/2}  K^{-1}\sigma_1^{-1/2}\sigma^{1/2}K\|_2\leq\\
&  
\frac{1}{\pi} T^{1/2}\left(C^f_{T}\right)^{1/2} (S_f^K(\rho||\sigma) - S_f^{K_1}(\rho_1||\sigma_1))^{1/2}  \ +\frac{4}{\pi T^{1/2}}(\|\Delta_{\sigma,\rho}\|+\|K\|)\ .
\end{align}
which is the analog of Theorem \ref{thm:mono} with a drfferent constant on the right, but the roles of $\rho$ and $\sigma$ interchanged there.  Note that since $V$ in unitary, $K$ is invertible if and only if $K_1$ is. Moreover, $\left(K^{-1}\right)^{-1}\rho_1 K^{-1}=\left(K_1^{-1}\right)^{-1}\rho_1 K_1^{-1}\otimes I$. Then using Lemma 2.2 in \cite{CV17-1} once more, we obtain
$$\|\rho^{1/2} -  \rho_1^{1/2}  K^{-1}\sigma_1^{-1/2}\sigma^{1/2}K\|_2\geq \frac{1}{2}\|K^*\sR_\sigma\left(\left(K_1^{-1}\right)^{-1}\rho_1 K_1^{-1}\right)K-\rho \|_1.$$
\end{proof}

%%%%%%%%%%%%%%
\subsection{Proof}\label{sec:proof-mono}
We are inspired by the proof of (\ref{eq:CV-mono}) in \cite{CV17-2}. And for completeness sake we provide all statements here as well. First, we recall Sharma's proof \cite{S14} of the monotonicity of the 
quasi relative entropies $S_f^K$ for operator convex $f$ and $K=K_1\otimes I_2$, and modify it accordingly for $K=K\otimes V$, where $V$ is unitary.
 
 \begin{proof}[Proof of Theorem~\ref{thm:original-mono} (Sharma '14)]
Define the operator $U$ mapping on $\cH=\cH_1\otimes\cH_2$ by
\begin{equation}\label{Ujdef}
U(X) = (X_1\otimes V)\rho_1^{-1/2}\rho^{1/2}\ .
\end{equation}
The adjoint operator on $\cH$  is given by
\begin{equation}\label{Ujcal}
U^*(Y) = \Tr_2(Y\rho^{1/2}(I_1\otimes V^*)) \rho_1^{-1/2}\ 
\end{equation}
for all $Y$ on $\cH $. 

Then note that $U$ is an isometry on $\cB(\cH_1)$
\begin{align}
\langle U(X_1), U(Y_1)\rangle&=\Tr\left(\rho^{1/2}\rho_1^{1/2}(X_1^*\otimes V^*)(Y_1\otimes V)\rho_1^{-1/2}\rho^{1/2} \right)\\
&=\Tr\left(\rho\left(\rho_1^{-1/2}X_1^*Y_1^*\rho_1^{-1/2}\otimes V^*V \right) \right)\\
&=\Tr\left(\rho_1^{1/2}X_1^*Y_1^*\rho_1^{-1/2} \right)\\
&=\langle X_1, Y_1\rangle\ .
\end{align}

Now observe that for all $X_1$ on $\cH_1$, 
\begin{align}
  U^* \Delta_{\sigma,\rho} U(X_1)&= \Tr_2\left(\sigma (X_1\rho_1^{-1/2}\otimes V)\rho^{1/2}\rho^{-1}\rho^{1/2}(I\otimes V^*)\rho_1^{-1/2}\right)\\
  &=\Tr(\sigma(X_1\rho_1^{-1}\otimes V V^*))\\
  &=\sigma_1 X_1\rho_1^{-1}\\
  &=\Delta_{\sigma_1,\rho_1}(X_1)\label{key} .
\end{align}

By the operator Jensen inequality
\begin{equation}\label{opjen}
 f\left(   U^* \Delta_{\sigma,\rho} U\right)\leq U^* f(\Delta_{\sigma,\rho}) U\ .
\end{equation}
Combining \eqref{key} and  \eqref{opjen}, and using  the fact that 
\begin{equation}\label{eq:UK}
U (K_1\rho_1^{1/2}) = K\rho^{1/2},
\end{equation}
we obtain 
\begin{eqnarray*}
S_f^{K_1}(\rho_1 || \sigma_1) &=& \langle K_1\rho_1^{1/2} , f(\Delta_{\sigma_1,\rho_1})  \left(K_1\rho_1^{1/2}\right)\rangle \\
&\leq &  \left \langle  U \left(K_1\rho_1^{1/2}\right) , f(\Delta_{\sigma,\rho})  U \left(K_1\rho_1^{1/2}\right)\right \rangle\\
&= &  \left  \langle K\rho^{1/2} ,  f(\Delta_{\sigma,\rho}) K\rho^{1/2}\right \rangle
= S_f^K(\rho||\sigma)\ .\\
\end{eqnarray*}
This proves the monotonicity theorem for the quasi relative entropy $S_f^K$ for every operator convex function $f$ and operator $K=K_1\otimes V$ for any unitary $V$. 
\end{proof}

We will use \cite[Lemma 2.1] {CV17-1}, the statement of which is the following:
\begin{lemma}[Carlen, Vershynina '17]\label{conlemA}
Let $U$ be a partial isometry embedding  a  Hilbert space $\cK$ into a Hilbert space $\cH$.
Let $B$ be an invertible  positive operator on $\cK$, $A$ be an invertible  positive operator on $\cH$, and suppose that 
$ U^*AU =   B$. Then for all $v\in \cK$,
\begin{equation}\label{resineq}
\langle v, U^* A^{-1} U v\rangle =  \langle v, B^{-1} v\rangle+ \langle w, A w\rangle\ ,
\end{equation}
where
\begin{equation}\label{resineq2}
w := UB^{-1}v - A^{-1}Uv\ .
\end{equation}
\end{lemma}

\begin{proof}[Proof of Theorem~\ref{thm:mono}] 

For an operator monotone decreasing function $f$, according to the integral representation (\ref{low}) the quasi-relative entropy $S_{f}^K$ can we written as
$$S_f^K(\rho\|\sigma)=-a\Tr(K^*\sigma K)-b\Tr(K\rho K^*)+\int_0^\infty\left(S_{(t)}^K(\rho||\sigma)-\frac{t}{t^2+1}\Tr(K^*K\rho)\right){\rm d}\mu_f(t)\ ,  $$
for $a\geq0$ and $b\in\bR$. Here
$$S_{(t)}^K(\rho|| \sigma)= \Tr(\sqrt{\rho}K^*(t\one+\Delta_{\sigma,\rho})^{-1}K\sqrt{\rho}).$$
Because $\Tr((K_1\otimes V) A (K_1\otimes V)^*)=\Tr(K_1 A_1 K_1)$,  it is clear that the difference between relative entropies can be written in terms of the  $S_{(t)}$-family, 
\begin{equation}\label{eq:int_rep}
S_f^K(\rho\|\sigma)-S_f^{K_1}(\rho_1\|\sigma_1)=\int_0^\infty\left(S_{(t)}^K(\rho||\sigma)-S_{(t)}^{K_1}(\rho_1||\sigma_1)\right){\rm d}\mu_f(t)\ .
\end{equation}

We apply  Lemma~\ref{conlemA} with $A := (t\one + \Delta_{\sigma,\rho})$,  $B = (t\one +\Delta_{\sigma_1,\rho_1})$ and 
$v := K_1\rho_1^{1/2}$, and with $U$ defined as above. The lemma's condition, $U^*AU = B$, follows from  \eqref{key} and the fact that $U^*U(X_1) = X_1$ for any $X_1$ on $\cH_1$. 

Therefore, applying Lemma~\ref{conlemA} with (\ref{eq:UK}), 
\begin{eqnarray}
S_{(t)}^K(\rho||\sigma) - S_{(t)}^{K_1}(\rho_1||\sigma_1) &=& \langle K\rho^{1/2},   (t\one+\Delta_{\sigma,\rho})^{-1} K\rho^{1/2}\rangle -  
 \langle K_1 \rho_1^{1/2},   (t\one+\Delta_{\sigma_1,\rho_1})^{-1}   K_1 \rho_1^{1/2}\rangle\nonumber\\
 &=&  \langle w_t, (t\one + \Delta_{\sigma,\rho}) w_t\rangle 
 \geq  t\|w_t\|^2, \label{eq:diff}
 \end{eqnarray}
 where,
 \begin{equation}\label{resolv}
w_t := U (t\one +\Delta_{\sigma_1,\rho_1})^{-1} (K_1\rho_1^{1/2}) -  (t\one + \Delta_{\sigma,\rho})^{-1}K\rho^{1/2}\ .
\end{equation}

{Notice that} by definition of $U$ (\ref{Ujdef}) and (\ref{eq:UK})
$$-w_t = U[t^{-1}\one  -  (t\one  +\Delta_{\sigma_1,\rho_1})^{-1}] (K_1\rho_1^{1/2}) -
 [t^{-1}\one  -  (t\one + \Delta_{\sigma,\rho})^{-1}]K\rho^{1/2}\ . $$
Since $U$ is an isometry on $\cB(\cH_1)$, 
$$\|w_t\| \leq  \| [t^{-1}\one -  (t\one +\Delta_{\sigma_1,\rho_1})^{-1}] (K_1\rho_1^{1/2}) \| + 
 \| [t^{-1}\one -  (t\one + \Delta_{\sigma,\rho})^{-1}]K\rho^{1/2}\|\ .$$
Since the  the modular operator is non-negative,  $0 \leq t^{-1}\one -  (t\one +\Delta_{\sigma_1,\rho_1})^{-1}  \leq t^{-1}\one$,
with the analogous estimate valid with $\Delta_{\sigma,\rho}$ in place of $ \Delta_{\sigma_1,\rho_1}$, Therefore,
\begin{equation}\label{newtwist}
\|w_t\| \leq 2 t^{-1} {\|K\|}\ .
\end{equation}

Now using the integral representation of the power function (recall that $\beta\in(0,1)$)
 $$
 X^{\beta} = \frac{\sin{\beta\pi}}{\pi} \int_0^\infty t^{\beta} \left(\frac{1}{t}\one  - \frac{1}{t+X}\right){\rm d}t,
 $$
 and (\ref{eq:UK}) once more, we conclude that
 \begin{equation}\label{eq:U-int}
 U(\Delta_{\sigma_1,\rho_1})^{\beta} (K_1\rho_1^{1/2})  -  (\Delta_{\sigma,\rho})^{\beta}K\rho^{1/2} = -  \frac{\sin{\beta\pi}}{\pi}\int_0^\infty t^{\beta}w_t{\rm d}t\ .
\end{equation}
 On the other hand,
\begin{eqnarray*}\label{eq:U-error}
 U(\Delta_{\sigma_1,\rho_1})^{\beta} (K_1\rho_1^{1/2})  -  (\Delta_{\sigma,\rho})^{\beta}K\rho^{1/2}&=& 
 U (\sigma_1^{\beta}K_1\rho_1^{1/2-\beta}) - \sigma^{\beta}K\rho^{1/2-\beta} \\
 &=& \sigma_1^{\beta}K\rho_1^{-\beta} \rho^{1/2} -\sigma^{\beta}K\rho^{1/2-\beta}\ .
 \end{eqnarray*}
 
 Combining the last two equalities, and taking the Hilbert space norm associated with $\cH$, for any $T_L, T_R>0$, 
\begin{eqnarray}\label{eq:REM1}
&&\|\sigma_1^{\beta}K\rho_1^{-\beta} \rho^{1/2} -\sigma^{\beta}K\rho^{1/2-\beta}\|_2 =
\frac{\sin{\beta\pi}}{\pi}\left\Vert \int_0^\infty  t^{\beta} w_t{\rm d}t\right\Vert_2  \nonumber\\
&\leq& \frac{\sin{\beta\pi}}{\pi}\int_0^{1/T_L} t^{\beta}  \|w_t\|_2{\rm d}t + \frac{\sin{\beta\pi}}{\pi} \int_{1/T_L}^{T_R} t^{\beta}  \|w_t\|_2{\rm d}t + \frac{\sin{\beta\pi}}{\pi}\left\| \int_{T_R}^\infty t^{\beta} w_t {\rm d}t \right\|_2\ .
\end{eqnarray}
Let us look at these three terms separately. The first term can be bounded using \eqref{newtwist}:
\begin{eqnarray}\label{eq:REM1-1}
&&\int_0^{1/T_L} t^{\beta}  \|w_t\|_2{\rm d}t \leq 2 \int_0^{1/T_L} t^{\beta -1}  {\rm d}t=\frac{2}{\beta}\frac{ {\|K\|}}{T_L^{\beta}}\ .
\end{eqnarray}

The third term in (\ref{eq:REM1}) can be bounded the following way: For any positive  operator $X> 0$,
$$
t^{\beta} \left(\frac{1}{t}\one  - \frac{1}{t+X}\right)   
\leq  t^{\beta} \left(\frac{1}{t} - \frac{1}{t+\|X\|}\right)\one  =  \frac{t^{\beta-1}\|X\|}{(t+\|X\|)}\one,
$$
and hence
$$
\int_{T}^\infty t^{\beta} \left(\frac{1}{t}\one - \frac{1}{t+X}\right) {\rm d}t \leq \|X\|^{\beta}\left(\int_{T/\|X\|}^\infty \frac{t^{\beta-1}}{1+t}  {\rm d}t\right) \one \leq \frac{\|X\|}{{(1-\beta)}T^{1-\beta}}\one\ .
$$
Since spectra of $\sigma_1$ and $\rho_1$ lie in the convex hulls of the spectra of $\sigma$ and $\rho$ respectively, it follows that $\|\Delta_{\sigma_1,\rho_1}\|  \leq \|\Delta_{\sigma,\rho}\|$.  Therefore, recalling the definition of $w_t$, we obtain
\begin{equation}\label{eq:third}
\left\| \int_{T_R}^\infty t^{\beta} w_t {\rm d}t \right\|_2  \leq \frac{2 \|\Delta_{\sigma,\rho}\|}{(1-\beta) \, T_R^{1-\beta}}\ .
\end{equation}

The second term can be bounded using Cauchy-Schwartz inequality and the fact that $f$ is regular, i.e. there is a constant $C^f_{T_L, T_R}$ such that $dt\leq C^f_{T_L, T_R}{\rm d}\mu_f(t)$ 
 for $t\in[1/T_L, T_R]$.\\
   {\bf Case 1:} $\beta\leq1/2$.
\begin{eqnarray}\label{eq:REM2-1}
\left(\int_{1/T_L}^{T_R} t^{\beta}  \|w_t\|_2{\rm d}t\right)^2&\leq& T_R\int_{1/T_L}^{T_R} t^{2\beta}  \|w_t\|_2^2{\rm d}t\nonumber\\
&\leq& T_RT_L^{1-2\beta}\int_{1/T_L}^{T_R} t \|w_t\|_2^2{\rm d}t\nonumber\\
&\leq&  T_RT_L^{1-2\beta}\int_{1/T_L}^{T_R} S_{(t)}(\rho||\sigma)-S_{(t)}(\rho_1\|\sigma_1) {\rm d}t\nonumber\\
&\leq&  T_RT_L^{1-2\beta}C^f_{T_L, T_R}\int_{1/T_L}^{T_R} S_{(t)}(\rho||\sigma)-S_{(t)}(\rho_1\|\sigma_1) {\rm d}\mu_f(t)\nonumber\\
&\leq& T_RT_L^{1-2\beta}\,C^f_{T_L, T_R}\, (S_f(\rho||\sigma)-S_f(\rho_1\|\sigma_1))\ .
\end{eqnarray}

 Therefore, combining (\ref{eq:REM1}), (\ref{eq:REM1-1}), (\ref{eq:third}), and (\ref{eq:REM2-1})  we have 
\begin{eqnarray*}\label{eq:bound}
\frac{\pi}{\sin{\beta\pi}}\|\sigma_1^{\beta}K\rho_1^{-\beta} \rho^{1/2} -\sigma^{\beta}K\rho^{1/2-\beta}\|_2
&\leq& \frac{2 {\|K\|}}{\beta T_L^{\beta}}+\frac{2 \|\Delta_{\sigma,\rho}\|}{(1-\beta) \, T_R^{1-\beta}}\\&&+T_R^{1/2}T_L^{1/2-\beta}\,\left(C^f_{T_L, T_R}\right)^{1/2}\, (S_f(\rho||\sigma)-S_f(\rho_1\|\sigma_1))^{1/2}\ .
\end{eqnarray*}

Taking $T_L:=T$ and  $T_R:=T^{\beta/(1-\beta)}$ we obtain
\begin{eqnarray*}
&&\frac{\pi}{\sin{\beta\pi}}\|\sigma_1^{\beta}K\rho_1^{-\beta} \rho^{1/2} -\sigma^{\beta}K\rho^{1/2-\beta}\|_2\\
&\leq& 2\left(\frac{ {\|K\|}}{\beta}+\frac{\|\Delta_{\sigma,\rho}\|}{1-\beta}\right)\frac{1}{T^{\beta}}+T^{\frac{1-\beta}{2}+\frac{\beta^2}{2(1-\beta)}}\,\left(C_{T, \beta}^f\right)^{1/2}\, (S_f(\rho||\sigma)-S_f(\rho_1\|\sigma_1))^{1/2}\ .
\end{eqnarray*}

{\bf Case 2:} $\beta>1/2$. 
\begin{eqnarray}
\left(\int_{1/T_L}^{T_R} t^{\beta}  \|w_t\|_2{\rm d}t\right)^2&\leq& T_R\int_{1/T_L}^{T_R} t^{2\beta}  \|w_t\|_2^2{\rm d}t\nonumber\\
&\leq& T_R^{2\beta}\int_{1/T_L}^{T_R} t \|w_t\|_2^2{\rm d}t\nonumber\\
&\leq&  T_R^{2\beta}\int_{1/T_L}^{T_R}( S_{(t)}(\rho||\sigma)-S_{(t)}(\rho_1\|\sigma_1) ){\rm d}t\nonumber\\
&\leq&  T_R^{2\beta}C^f_{T_L, T_R}\int_{1/T_L}^{T} (S_{(t)}(\rho||\sigma)-S_{(t)}(\rho_1\|\sigma_1) ){\rm d}\mu_f(t)\nonumber\\
&\leq& T_R^{2\beta}\,C^f_{T_L, T_R}\, (S_f(\rho||\sigma)-S_f(\rho_1\|\sigma_1))\label{eq:REM1-2}\ .
\end{eqnarray}

Therefore, combining (\ref{eq:REM1}), (\ref{eq:REM1-1}), (\ref{eq:third}), and (\ref{eq:REM1-2}) we have
\begin{eqnarray*}\label{eq:bound}
\frac{\pi}{\sin{\beta\pi}} \|\sigma_1^{\beta}K\rho_1^{-\beta} \rho^{1/2} -\sigma^{\beta}K\rho^{1/2-\beta}\|_2&\leq& \frac{2 {\|K\|}}{\beta T_L^{\beta}}+\frac{2 \|\Delta_{\sigma,\rho}\|}{(1-\beta) \, T_R^{1-\beta}}\\&&+T_R^{\beta}\,\left(C^f_{T, \beta}\right)^{1/2}\, (S_f(\rho||\sigma)-S_f(\rho_1\|\sigma_1))^{1/2}\ .
\end{eqnarray*}

Taking $T_L:=T^{(1-\beta)/\beta}$ and  $T_R:=T$ we obtain
\begin{eqnarray*}\label{eq:bound2}
&&\frac{\pi}{\sin{\beta\pi}}\|\sigma_1^{\beta}K\rho_1^{-\beta} \rho^{1/2} -\sigma^{\beta}K\rho^{1/2-\beta}\|_2\\
&\leq& 2\left(\frac{ {\|K\|}}{\beta}+\frac{\|\Delta_{\sigma,\rho}\|}{1-\beta}\right)\frac{1}{T^{1-\beta}}+T^{\beta}\,\left(C^f_{T, \beta}\right)^{1/2}\, (S_f(\rho||\sigma)-S_f(\rho_1\|\sigma_1))^{1/2}\ .
\end{eqnarray*}
\end{proof}

%%%%%%%
\subsection{Condition for equality}

Corollary~\ref{optimcl} and the proof of Theorem~\ref{thm:mono} give a condition on the equality in the monotonicity inequality.
\begin{corollary}\label{cor:eq-mono}
Let $f$ be a regular function. The equality in the monotonicity inequality
\begin{equation}\label{eq:mono-eq}
 S_f^K(\rho||\sigma)-S_f^{K_1}(\rho_1\|\sigma_1)=0,
 \end{equation}
 holds if and only if for all $\beta\in\mathbb{C}$ the following holds:
 \begin{equation}\label{eq:beta}
\sigma_1^{\beta}K\rho_1^{-\beta}  = \sigma^{\beta}K\rho^{-\beta}.
\end{equation}
\end{corollary}
\begin{proof}
If the equality in the monotonicity inequality holds, then (\ref{eq:beta}) holds for all $\beta\in(0,1)$  following the Corollary~\ref{optimcl}. Since  for any positive matrix $X$, the map $\beta \rightarrow X^\beta$ is an entire analytic function, this identity holds for all $\beta\in\mathbb{C}$.

The other way, suppose (\ref{eq:beta}) holds for all $\beta\in\mathbb{C}$. Then from  (\ref{eq:U-error}) in the proof of the Theorem~\ref{thm:mono}, we have that
$$U(\Delta_{\sigma_1,\rho_1})^{\beta} (K_1\rho_1^{1/2})  -  (\Delta_{\sigma,\rho})^{\beta}K\rho^{1/2}=0\ , $$
for all $\beta\in\mathbb{C}$.  Let us use the following Taylor series expansion 
$$\frac{1}{t+x}=\sum_{n=0}^\infty \frac{(-1)^n}{t^{n+1}}x^n\ . $$
Then, using the above two equalities, we obtain that for all $t\geq 0$,
$$w_t=U (t\one +\Delta_{\sigma_1,\rho_1})^{-1} (K_1\rho_1^{1/2}) -  (t\one + \Delta_{\sigma,\rho})^{-1}K\rho^{1/2}=0\ . $$
From (\ref{eq:diff}) this implies that $S_{(t)}^K(\rho||\sigma) - S_{(t)}^{K_1}(\rho_1||\sigma_1) =0$, and therefore, the (\ref{eq:mono-eq}) is satisfied, following  the integral representation (\ref{eq:int_rep}).

\end{proof}

%%%%%%%%%%%%%%%%%
\section{Joint convexity of the quasi-relative entropy}\label{sec:Joint}

As it was shown by Petz \cite{OP93, P86-1} or \cite[Theorem 2]{P10}, the quasi relative entropy is jointly convex in $\rho$ and $\sigma$. Here is another elegant proof of a joint convexity.

\begin{proposition}\label{prop:convex}
For an operator monotone decreasing function $f$, and any operator $K$, quasi-entropy $S_f^K(\rho\| \sigma)$ is jointly convex in $\rho, \sigma>0$. In other words, for $\rho=\sum_jp_j\rho_j$ and $\sigma=\sum_j p_j\sigma_j$,
$$0\leq  \sum_j p_j S_f^K(\rho_j\| \sigma_j)-S_f^K(\rho\|\sigma). $$
Moreover, the equality in the joint convexity holds if and only if 
$$(\Delta_{\sigma,\rho}+tI)^{-1}(K)=(\Delta_{\sigma_j,\rho_j}+tI)^{-1}(K),  \ \text{for all}\ j \ \text{and for all} \ t>0.$$
\end{proposition}
\begin{proof}
Using (\ref{low}) representation of operator monotone decreasing function, we obtain 
\begin{align}
S_f^K(\rho\| \sigma)=&-a \Tr K^*\sigma K-b\Tr K^*K\rho+\int_0^\infty \left\{\Tr \sqrt{\rho}K^*\frac{1}{\Delta_{\sigma,\rho}+t}(K\sqrt{\rho})-\frac{t}{t^2+1}\Tr K^*K\rho \right\} {\rm d}\mu_f(t)\nonumber\\
=& -a \Tr K^*\sigma K-b\Tr K^*K\rho+\int_0^\infty \left\{\Tr \rho K^*\frac{1}{L_\sigma+tR_\rho}(K\rho)-\frac{t}{t^2+1}\Tr K^*K\rho \right\} {\rm d}\mu_f(t)\nonumber\\
=&-a \Tr K^*\sigma K-b\Tr K^*K\rho+\int_0^\infty \left\{\Tr \rho K^*\frac{1}{\Delta_{\sigma,\rho}+t}(K)-\frac{t}{t^2+1}\Tr K^*K\rho \right\} {\rm d}\mu_f(t).\label{eq:S_f}
\end{align}
The joint convexity follows immediately from that of the map $(Y, A, B)\rightarrow \Tr Y^* \frac{1}{L_B+tR_A}(Y)$, which was proved in \cite{R05}.

Note that equality in the joint convexity holds if and only if it holds for the first term in the integrand.
\end{proof}

The next theorem uses monotonicity inequality in Theorem \ref{thm:mono} and  provides a strengthening of the convexity inequality.
\begin{theorem}\label{thm:joint-new}
Let $\cH=\cH_1\otimes\cH_2$, $\beta\in(0,1)$, $K\in\cB(\cH)$. Let $f$ be a $C_{T, \beta}^f$-regular function. Then for states $\rho=\sum_jp_j\rho_j$ and $\sigma=\sum_j p_j\sigma_j$ (with $p_j>0$ and $\sum_j p_j=1$), we have
\begin{eqnarray}\label{eq:thm-conv}
&&\frac{\pi}{\sin{\beta\pi}}\sum_j p_j^{1/2} \|\sigma^{\beta}K\rho^{-\beta} \rho_j^{1/2} - \sigma_j^{\beta}K\rho_j^{1/2-\beta}\|_2\nonumber\\
&\leq& 2\left(\frac{ {\|K\|}}{\beta}+\frac{\sum_jp_j^{-1}\|\rho_j^{-1}\|}{1-\beta}\right)\frac{1}{T^{\alpha_1(\beta)}}+T^{\alpha_2(\beta)}\,\left(C_{T, \beta}^f\right)^{1/2}\, \left(\sum_j p_j S_f^K(\rho_j\| \sigma_j) - S_f^K({\rho}\| {\sigma})\right)^{1/2} \ ,
\end{eqnarray}
where 
$$\alpha_1(\beta)=\left\{\begin{matrix}
\beta &\text{when }\beta\leq1/2\\
1-\beta &\text{when }\beta\geq1/2.
\end{matrix}\right. \ \ \text{and} \ \ 
\alpha_2(\beta)=\left\{\begin{matrix}
\frac{1-\beta}{2}+\frac{\beta^2}{2(1-\beta)} &\text{when }\beta\leq1/2\\
\beta&\text{when }\beta\geq1/2.
\end{matrix}\right. $$
\end{theorem}
\begin{proof}
 Let us form the following quantum-classical states 
$$\overline{\rho}:=\sum_j p_j \rho_j\otimes\ket{j}\bra{j}_X, $$
$$ \overline{\sigma}:=\sum_j p_j \sigma_j\otimes\ket{j}\bra{j}_X.$$
Then $\overline{\rho}_1=\rho$ and $\overline{\sigma}_1=\sigma$, i.e. $$S_f^K(\overline{\rho}_1\|\overline{\sigma}_1)=S_f^K({\rho}\| {\sigma}).$$

Let us use (\ref{eq:S_f}) for the expression of the quasi-relative entropy. There we see that all but one term are linear in $\rho$ and $\sigma$. For the term with the modular operator, note that for any $A=\sum_j A_j\otimes\ket{j}\bra{j}_X$ we have
\begin{equation}\label{eq:delta}
\Delta_{\overline{\sigma},\overline{\rho}}(A)=\sum_j\Delta_{\sigma_j,\rho_j}(A_j)\otimes\ket{j}\bra{j}_X. 
\end{equation}
Since $K=K\otimes I_X$, we may apply this to a term in (\ref{eq:S_f}), and obtain
$$ \Tr \rho K^*\frac{1}{\Delta_{\overline{\sigma},\overline{\rho}}+t}(K)=\sum_j p_j \Tr\left\{ \rho K^*\frac{1}{\Delta_{\sigma_j,\rho_j}+t}(K)\right\}.$$
Therefore, from (\ref{eq:S_f})
$$S_f^K(\overline{\rho}\| \overline{\sigma})=\sum_j p_j S_f^K(\rho_j\| \sigma_j).$$

Note that 
$$\|\Delta_{\overline{\sigma}, \overline{\rho}}\|\leq \|\overline{\sigma}\|\|\overline{\rho}^{-1}\|\leq \|\overline{\rho}^{-1}\|\leq \sum_j p_j^{-1}\|\rho_j^{-1}\|. $$

Applying Theorem~\ref{thm:mono} will give us the right-hand side of (\ref{eq:thm-conv}). The left-hand side results from the following identity
$$\overline{\sigma}_1^{\beta}K\overline{\rho}_1^{-\beta}\overline{\rho}^{1/2}-\overline{\sigma}^{\beta}K\overline{\rho}^{1/2-\beta}=\sum_j p_j^{1/2} \left\{\sigma^{\beta}K\rho^{-\beta}\rho_j^{1/2}-\sigma_j^{\beta}K\rho_j^{1/2-\beta} \right\}\otimes \ket{j}\bra{j}_X.$$
Taking the Hilbert-Schmidt norm (\ref{eq:HS-norm}) on both sides will result in the correct left-hand side.  
\end{proof}

\begin{corollary}\label{cor:conv} Let $\beta\in(0,1)$ and $f$ be a $C_{T, \beta}^f$-regular function. Suppose that $C^f_{T, \beta} \leq C\, T^{2c}$ for some $c, C>0$. Let  $\rho=\sum_jp_j\rho_j$ and $\sigma=\sum_j p_j\sigma_j$ (with $p_j>0$ and $\sum_j p_j=1$). Then there is an explicitly computable  constant 
$M$ depending only on the smallest non-zero eigenvalues of $\{\rho_j\}_j$,  {$\|K\|$}, $\{p_j^{-1}\}_j$, $\beta$, $C$ and $c$, such that,\\
\begin{eqnarray}\label{eq:cor-conv}
\sum_j p_j^{1/2} \|\sigma^{\beta}K\rho^{-\beta} \rho_j^{1/2} - \sigma_j^{\beta}K\rho_j^{1/2-\beta}\|_2\leq M 
\left(\sum_j p_j S_f^K(\rho_j\| \sigma_j) - S_f^K({\rho}\| {\sigma})\right)^{\alpha(\beta) }\ ,
\end{eqnarray}
where 
$$\alpha(\beta)=\left\{\begin{matrix}
\frac{\beta(1-\beta)}{1+2c(1-\beta)} &\text{when }\beta\leq1/2\\
\frac12  \frac{1-\beta}{ 1+c}&\text{when }\beta\geq1/2.
\end{matrix}\right. $$

In particular, for $\beta=1/2,$
\begin{equation}
\sum_j p_j^{1/2} \|\sigma^{1/2}K\rho^{-1/2} \rho_j^{1/2} - \sigma_j^{1/2}K\|_2\leq M\left(\sum_j p_j S_f^K(\rho_j\| \sigma_j) - S_f^K({\rho}\| {\sigma})\right)^{1/4(c+1)}\ .
\end{equation} 
\end{corollary}

\subsection{Condition for equality}
From Corollary~\ref{cor:eq-mono} and the proof of Theorem~\ref{thm:joint-new} we obtain the condition on the equality in the joint convexity inequality.
\begin{corollary}
An equality in the joint convexity inequality
$$\sum_j p_j S_f^K(\rho_j\| \sigma_j) = S_f^K({\rho}\| {\sigma}), $$
holds if and only if, for all $j$ and all $\beta\in\mathbb{C}$
$$\sigma^{\beta}K\rho^{-\beta}  = \sigma_j^{\beta}K\rho_j^{-\beta}\ . $$
\end{corollary}

%%%%%%%%%%%%
%%%%%%%%%%%%%%%%%%%%
\section{Operator inequalities}\label{sec:ssa}

In this section consider a tri-partite Hilbert space $\cH=\cH_A\otimes\cH_B\otimes\cH_C$. Let $\rho=\rho_{ABC}$ be a state on $\cH$. Then the strong subadditivity of quantum entropy is the following statement:
\begin{theorem}[Lieb, Ruskai '73]\label{thm:ssa-reg}
For $\rho_{ABC}$ a state on $\cH_{ABC}$, it holds that
\begin{equation}
0\leq S(\rho_{AB})+S(\rho_{BC})-S(\rho_{ABC})-S(\rho_B)\ .
\end{equation}
\end{theorem}
This theorem was proved by Lien and Ruskai \cite{LR73}, using  Lieb's theorem that was proved in \cite{L73}. The theorem in a von Neumann algebra setting was done by Narnhofer and Thirring in \cite{NT85}.

 Let $\sigma=\sigma_{AB}\otimes I_C$ be a state on $\cH$. Let $f$ be an operator monotone decreasing function. In \cite{R12}  Rusaki, building on the work of Kim \cite{K12}, showed that the following operator on $\cH_C$ is positive semi-definite
\begin{equation}\label{eq:Ruskai}
0\leq  \Tr_{AB}\left([f(L_{\sigma_{AB}}R^{-1}_{\rho_{ABC}})-f(L_{\sigma_{B}} R^{-1}_{\rho_{BC}})]\rho_{ABC}\right)\ .
\end{equation}
In particular, taking $\sigma_{AB}=\rho_{AB}$, and $f(x)=-\log(x)$, reduces to an operator inequality of Kim's \cite{K12}
\begin{equation}\label{eq:Kim}
0\leq \Tr_{AB}[\log\rho_{ABC}-\log\rho_{AB}-\log\rho_{BC}+\log\rho_B ]\rho_{ABC}  \ .
\end{equation}

We prove the following sharpening of (\ref{eq:Ruskai}) inequality.

\begin{theorem}\label{thm:ssa} Let $\cH=\cH_A\otimes\cH_B\otimes\cH_C$, $\beta\in(0,1)$ and $\rho=\rho_{ABC}$, $\sigma=\sigma_{AB}\otimes I_C$. Let $f$ be a $C_{T, \beta}^f$-regular function. Suppose that $C^f_{T, \beta} \leq C\, T^{2c}$ for some $c, C>0$. Then there is an explicitly computable positive constant 
$N$ depending only on the smallest non-zero eigenvalue of $\rho$,  {$\|K\|$}, $\beta$, $C$ and $c$, such that, the following operator inequality holds on $\cH_C$
\begin{align}
N [\Tr_{AB}(P_{ABC}(\rho,\sigma)P^*_{ABC}(\rho,\sigma))]^{1/\alpha(\beta)}\leq \Tr_{AB}\left([f(L_{\sigma_{AB}}R^{-1}_{\rho_{ABC}})-f(L_{\sigma_{B}} R^{-1}_{\rho_{BC}})]\rho_{ABC}\right)\ .\label{eq:ssa}
\end{align}
where
$$P_{ABC}(\rho,\sigma):=\sigma_{B}^{\beta}\rho_{BC}^{-\beta} \rho_{ABC}^{1/2} -\sigma_{AB}^{\beta}\rho_{ABC}^{1/2-\beta}\ , $$
and $$\alpha(\beta)=\left\{\begin{matrix}
\frac{\beta(1-\beta)}{1+2c(1-\beta)} &\text{when }\beta\leq1/2\\
\frac12  \frac{1-\beta}{ 1+c}&\text{when }\beta\geq1/2.
\end{matrix}\right. $$
\end{theorem}
The proof of this theorem is given in Section~\ref{proof:ssa}. Note that since the left-hand side of (\ref{eq:ssa}) is positive semi-definite on $\cH_C$, this theorem implies result in \cite{R12}, i.e. (\ref{eq:Ruskai}).

 In the proof of Theorem~\ref{thm:ssa}, if one takes $\rho_{ABC}=\rho_{AB}\otimes I_C$ and $\sigma_{ABC}$ (i.e. in (\ref{eq:D})), we would obtain
 
 \begin{theorem}\label{thm:ssa2}
  Let $\cH=\cH_A\otimes\cH_B\otimes\cH_C$, $\beta\in(0,1)$ and $\rho=\rho_{AB}\otimes I_C$, $\sigma_{ABC}$ be a state on $\cH$. Let $f$ be a $C_{T, \beta}^f$-regular function. Suppose that $C^f_{T, \beta} \leq C\, T^{2c}$ for some $c, C>0$. Then there is an explicitly computable positive constant 
$N$ depending only on the smallest non-zero eigenvalue of $\rho$,  {$\|K\|$}, $\beta$, $C$ and $c$, such that, the following operator inequality holds on $\cH_C$
 $$N[\Tr_{AB}(Q^*_{ABC}(\rho_{AB},\sigma_{ABC})Q_{ABC}(\rho_{AB},\sigma_{ABC})]^{1/\alpha(\beta)}\leq \Tr_{AB}\left([f(L_{\sigma_{ABC}}R^{-1}_{\rho_{AB}})-f(L_{\sigma_{BC}} R^{-1}_{\rho_{B}})]\rho_{AB}\right)\ , $$
where
$$Q_{ABC}(\rho_{AB},\sigma_{ABC}):=\sigma_{BC}^{\beta}\rho_{B}^{-\beta} \rho_{AB}^{1/2} -\sigma_{ABC}^{\beta}\rho_{AB}^{1/2-\beta}\ . $$
 \end{theorem}

 Let us take the function $\tilde{f}(x)=xf(1/x)$. By \cite[Theorem V.2.9]{B97} this function is operator monotone decreasing if and only if $f$ is. Taking $\tilde{f}$ instead of $f$ in Theorem~\ref{thm:ssa} and \ref{thm:ssa2} and interchanging the roles of $\rho$ and $\sigma$, leads to

\begin{corollary}\label{cor:ssa1} With the same conditions and notations as in Theorem~\ref{thm:ssa2}, we have
 $$N [\Tr_{AB}(P_{ABC}(\sigma,\rho)P^*_{ABC}(\sigma,\rho))]^{1/\alpha(\beta)}\leq\Tr_{AB}\left(\rho_{AB}[f(L_{\rho_{AB}}^{-1}R_{\sigma_{ABC}})-f(L_{\rho_{B}}^{-1}R_{\sigma_{BC}})]\right) \ . $$ 

\end{corollary}

  \begin{corollary}\label{cor:ssa3} With the same conditions and notations as in Theorem~\ref{thm:ssa}, we have
$$N[\Tr_{AB}(Q^*_{ABC}(\sigma_{AB},\rho_{ABC})Q_{ABC}(\sigma_{AB},\rho_{ABC})]^{1/\alpha(\beta)}\leq \Tr_{AB}\left(\rho_{ABC}[f(L^{-1}_{\rho_{ABC}}R_{\sigma_{AB}})-f(L^{-1}_{\rho_{BC}} R_{\sigma_{B}})]\right)\ . $$
\end{corollary}

%%%%%%%%%%
\subsection{Proof}\label{proof:ssa}
\begin{proof}[Proof of Theorem~\ref{thm:ssa}] We are inspired by the proof of Ruskai \cite{R12}, which we provide in our case in all detail for the completeness sake. 

In the monotonicity inequality Corollary~\ref{optimcl}, let us consider $\rho=\rho_{ABC}$, $\sigma=\sigma_{ABC}$, $\cH_1=\cH_{BC}$, $\cH_2:=\cH_A$, and $K_{ABC}=I_A\otimes K_{BC}$ . Then (\ref{eq:cor}) is equivalent to
\begin{align}\label{eq:ssa-1}
N\|\sigma_{BC}^{\beta}K_{BC}\rho_{BC}^{-\beta} \rho_{ABC}^{1/2} -\sigma_{ABC}^{\beta}K_{BC}\rho_{ABC}^{1/2-\beta}\|_2^{1/\alpha(\beta)}\leq 
S_f^{K_{BC}}(\rho_{ABC}||\sigma_{ABC})-S_f^{K_{BC}}(\rho_{BC}\|\sigma_{BC})\ .
\end{align}

Let us consider the difference on right-hand side in the above inequality.
\begin{align}
D:=& S_f^{K_{BC}}(\rho_{ABC}||\sigma_{ABC})-S_f^{K_{BC}}(\rho_{BC}\|\sigma_{BC})\nonumber\\
&=\Tr_{ABC}(\rho^{1/2}_{ABC}K^*_{BC}f(\Delta_{\sigma_{ABC},\rho_{ABC}})(K_{BC}\rho^{1/2}_{ABC}))-\Tr_{ABC}(\rho^{1/2}_{BC}K^*_{BC}f(\Delta_{\sigma_{BC},\rho_{BC}})(K_{BC}\rho^{1/2}_{BC}))\nonumber\\
&=\Tr_{ABC}(K^*_{BC}f(\Delta_{\sigma_{ABC},\rho_{ABC}})(K_{BC}\rho_{ABC}))-\Tr_{ABC}(K^*_{BC}f(\Delta_{\sigma_{BC},\rho_{BC}})(K_{BC}\rho_{BC}))\label{eq:D}
\end{align}
In the first equality we used the definition of the quasi-relative entropy and the fact that there is no dependence on $A$ in the second term. In the second equality we used the definition of the modular operator $\Delta_{A,B}=L_AR_B^{-1}$.

Consider a special case when $\sigma_{ABC}=\sigma_{AB}\otimes I_C$. Then
\begin{align}
D=\Tr_{ABC}(K^*_{BC}f(\Delta_{\sigma_{AB},\rho_{ABC}})(K_{BC}\rho_{ABC}))-\Tr_{ABC}(K^*_{BC}f(\Delta_{\sigma_{B},\rho_{BC}})(K_{BC}\rho_{BC}))\ .
\end{align}
Choose $K_{BC}=I_B\otimes K_C$. Then $K=I_{AB}\otimes K_C$ commutes with $\sigma_{ABC}=\sigma_{AB}\otimes I_C.$ Therefore,
$$D= \Tr_{ABC}(K^*_CK_{C}f(\Delta_{\sigma_{AB},\rho_{ABC}})(\rho_{ABC}))-\Tr_{ABC}(K_C^*K_{C}f(\Delta_{\sigma_{B},\rho_{BC}})(\rho_{BC}))\ .$$
Furthermore, take $K_C=\ket{\phi}\bra{\phi}_C$ to be a projector onto a vector $\ket{\phi}$. Then 
\begin{align}
D=\bra{\phi}_C\Tr_{AB}\left(f(\Delta_{\sigma_{AB},\rho_{ABC}})(\rho_{ABC})-f(\Delta_{\sigma_{B},\rho_{BC}})(\rho_{BC})\right)\ket{\phi}_C\ .
\end{align}
Bringing this expression back into (\ref{eq:ssa-1}), we obtain 
\begin{align}
&N\|\sigma_{B}^{\beta}\ket{\phi}\bra{\phi}_C\rho_{BC}^{-\beta} \rho_{ABC}^{1/2} -\sigma_{AB}^{\beta}\ket{\phi}\bra{\phi}_C\rho_{ABC}^{1/2-\beta}\|_2^{1/\alpha(\beta)}\leq \nonumber \\
&
\bra{\phi}_C\Tr_{AB}\left([f(\Delta_{\sigma_{AB},\rho_{ABC}})-f(\Delta_{\sigma_{B},\rho_{BC}})](\rho_{ABC})\right)\ket{\phi}_C\ .\label{eq:ssa-0}
\end{align}
Denoting 
$$P_{ABC}(\rho,\sigma):=\sigma_{B}^{\beta}\rho_{BC}^{-\beta} \rho_{ABC}^{1/2} -\sigma_{AB}^{\beta}\rho_{ABC}^{1/2-\beta}\ , $$
we calculate the norm on the left-hand side by definition (\ref{eq:HS-norm})
\begin{align}
\|\sigma_{B}^{\beta}\ket{\phi}\bra{\phi}_C\rho_{BC}^{-\beta} \rho_{ABC}^{1/2} -\sigma_{AB}^{\beta}\ket{\phi}\bra{\phi}_C\rho_{ABC}^{1/2-\beta}\|_2&= \|\ket{\phi}\bra{\phi}_C[\sigma_{B}^{\beta}\rho_{BC}^{-\beta} \rho_{ABC}^{1/2} -\sigma_{AB}^{\beta}\rho_{ABC}^{1/2-\beta}]\|_2\\
&=\Tr_{ABC}(\ket{\phi}\bra{\phi}_C P_{ABC}(\rho,\sigma)P^*_{ABC}(\rho,\sigma))\\
&=\bra{\phi}_C [\Tr_{AB}(P_{ABC}(\rho,\sigma)P^*_{ABC}(\rho,\sigma))]\ket{\phi}_C\ .
\end{align}
Therefore, we have for all $\ket{\phi}_C$
$$ \bra{\phi}_C [\Tr_{AB}(P_{ABC}(\rho,\sigma)P^*_{ABC}(\rho,\sigma))]\ket{\phi}_C\leq \bra{\phi}_C\Tr_{AB}\left([f(\Delta_{\sigma_{AB},\rho_{ABC}})-f(\Delta_{\sigma_{B},\rho_{BC}})](\rho_{ABC})\right)\ket{\phi}_C\ .$$
Since the above inequality holds for all $\ket{\phi}_C$ in $\cH_C$, we have the following operator inequality
$$N[\Tr_{AB}(P_{ABC}(\rho,\sigma)P^*_{ABC}(\rho,\sigma))]^{1/\alpha(\beta)}\leq \Tr_{AB}\left([f(L_{\sigma_{AB}}R^{-1}_{\rho_{ABC}})-f(L_{\sigma_{B}} R^{-1}_{\rho_{BC}})](\rho_{ABC})\right)\ . $$
\end{proof}

%%%%%%%%%%
\subsection{Condition for equality}
From the proof above it is evident that in Theorems~\ref{thm:ssa} and \ref{thm:ssa2} and Corollaries~\ref{cor:ssa1} and \ref{cor:ssa3} the right-hand side is zero if and only if for all $\beta\in(0,1)$, $P_{ABC}(\rho,\sigma)=0$ or $Q_{ABC}(\rho, \sigma)=0$. Let us take an example
\begin{proposition}\label{prop:eq}
Let $\cH=\cH_A\otimes\cH_B\otimes\cH_C$, $\beta\in(0,1)$ and $\rho=\rho_{ABC}$, $\sigma=\sigma_{AB}\otimes I_C$. Let $f$ be a $C_{T, \beta}^f$-regular function.
The equality 
\begin{equation}\label{eq:eq1}
\Tr_{AB}([f(L_{\sigma_{AB}}R^{-1}_{\rho_{ABC}})\rho_{ABC})=\Tr_B(f(L_{\sigma_{B}} R^{-1}_{\rho_{BC}})]\rho_{BC}) 
\end{equation}
holds if and only if
\begin{equation}\label{eq:eq2}
\sigma_{B}^{\beta}\rho_{BC}^{-\beta} =\sigma_{AB}^{\beta}\rho_{ABC}^{-\beta}\, ,\ \ \text{ for all }\beta\in(0,1)\ .
\end{equation}
Moreover, (\ref{eq:eq1}) is equivalent to the Petz's recovery condition
\begin{equation}\label{eq:eq3}
\sR_{\rho_{ABC}}(\sigma_B)=\sigma_{AB}\ .
\end{equation}
\end{proposition} 
\begin{proof}
From Theorem~\ref{thm:ssa}, it is clear that if (\ref{eq:eq1}) is satisfied, the following holds
$$0=\Tr_{AB}(P_{ABC}P_{ABC}^*).$$
Then for all $\ket{\phi}_C$ in $\cH_C$,
$$0=\Tr(P_{ABC}P^*_{ABC}\ket{\phi}_C\bra{\phi}_C )\ .$$
The operator inside trace is positive semi-definite, since it can written as $A^*A$ for $A=\ket{\phi}_C\bra{\phi}_CP_{ABC}$. Therefore, the operator itself is zero, i.e. for all $\ket{\phi}_C$
$$0=\ket{\phi}_C\bra{\phi}_CP_{ABC}\ . $$
Since this holds for all $\ket{\phi}_C$, the operator $P_{ABC}$ is zero, leading to the required (\ref{eq:eq2}). Taking $\beta=1/2$, equality (\ref{eq:eq2}) leads to (\ref{eq:eq3}).

On the other side, if (\ref{eq:eq2}) is satisfied, it means that in particular, Petz's recovery map recovers both $\rho_{BC}$ and $\sigma_B$ perfectly, i.e. (\ref{eq:eq3}) holds. Therefore, by the equivalence of the recovery of both states and the saturation of the monotonicity inequality, we have that $D$ in (\ref{eq:D}) is zero, i.e. (\ref{eq:eq1}) is satisfied.
\end{proof}

%%%%%%%%%%%%%%%
%%%%%%%%%%%%%%%%
\section{Logarithmic function}\label{sec:log}

Let us take $f(x)=-\log(x)$ and $K=I$. We may explicitly calculate the power $\alpha(\beta)$.  From Example~\ref{ex-log} we have that ${\rm d}\mu(x) = {\rm d}x$. Therefore, in Corollary~\ref{cor:conv} and Theorem~\ref{thm:ssa} the constants are: $c=0$, $C=1$, and
$$\alpha(\beta)=\left\{\begin{matrix}
{\beta(1-\beta)} &\text{when }\beta\leq1/2\\
  ({1-\beta})/{2}&\text{when }\beta\geq1/2.
\end{matrix}\right. \ .$$
Moreover, from \cite[Corollary 5.1]{CV17-2} 
$$N=\left\{\begin{matrix}
\left(\frac{\pi(1-2\beta+2\beta^2)\beta}{\sin{\beta\pi}}\right)^\frac{1}{\beta(1-\beta)}\left( {\|K\|}+\frac{\beta}{1-\beta}D \right)^{-\frac{1-2\beta+2\beta^2}{\beta(1-\beta)}}2^{-\frac{1-2\beta+2\beta^2}{\beta(1-\beta)}}\left(\frac{1-2\beta+2\beta^2}{2(1-\beta)}\right)^{-2} &\text{when }\beta\leq1/2\\
  \left(\frac{\pi\beta(1-\beta)}{\sin{\beta\pi}}\right)^\frac{2}{1-\beta}\left(\frac{1-\beta}{\beta} {\|K\|}+D\right)^{-\frac{2\beta}{1-\beta}}2^{-\frac{2\beta}{1-\beta}}\beta^{-2}&\text{when }\beta\geq1/2.
\end{matrix}\right. \ ,$$
where $D=\sum_jp_j^{-1}\|\rho_j^{-1}\|$ in Corollary~\ref{cor:conv-log} below, $D=\|\Delta_{\rho_A\otimes\rho_{BC},\rho_{ABC}}\|$ in Corollary~\ref{cor:ssa-log}, and $D=\|\Delta_{\sigma,\rho}\| $ in Corollary~\ref{cor:op-log}.

Taking $K=I$ in Corollary~\ref{cor:conv}, we obtain the following sharpening of the joint convexity of quantum relative entropy.
\begin{corollary}[Joint convexity]\label{cor:conv-log}
 Let  $\rho=\sum_jp_j\rho_j$ and $\sigma=\sum_j p_j\sigma_j$ (with $p_j>0$ and $\sum_j p_j=1$). With $\alpha(\beta)$ and $N$ defined above, we have
 \begin{eqnarray}\label{eq:cor-conv}
N\left(\sum_j p_j^{1/2} \|\sigma^{\beta}\rho^{-\beta} \rho_j^{1/2} - \sigma_j^{\beta}\rho_j^{1/2-\beta}\|_2\right)^{1/\alpha(\beta)}\leq 
\sum_j p_j S(\rho_j\| \sigma_j) - S({\rho}\| {\sigma})\ .
\end{eqnarray}
In particular, for $\beta=1/2,$
\begin{equation}
N\left(\sum_j p_j^{1/2} \|\sigma^{1/2}\rho^{-1/2} \rho_j^{1/2} - \sigma_j^{1/2}\|_2\right)^{4}\leq \sum_j p_j S(\rho_j\| \sigma_j) - S({\rho}\| {\sigma})\ .
\end{equation} 
\end{corollary}
From the last inequality, we have that if the joint convexity is saturated, i.e.
$$\sum_j p_j S(\rho_j\| \sigma_j) = S({\rho}\| {\sigma}), $$
then for all $j$ and all $\beta\in(0,1)$
$$\sigma^{\beta}\rho^{-\beta}  = \sigma_j^{\beta}\rho_j^{-\beta}\ . $$

From the sharpening of the monotonicity inequality,  Corollary~\ref{cor:mono1}, or the previous result (\ref{eq:cor}), we obtain the sharpening of the strong subadditivity when taking $\sigma_{ABC}=\rho_{AB}\otimes\rho_{C}$ and tracing out system $A$.

\begin{corollary}[Strong subadditivity]\label{cor:ssa-log} For $\rho:=\rho_{ABC}$ a state on $\cH_{ABC}$ and $\beta\in(0,1)$, it holds that
$$N 
\|\rho_B^{\beta}\otimes\rho_{C}^{\beta}\,\rho_{BC}^{-\beta} \rho^{1/2} - \rho_{AB}^{\beta}\otimes\rho_{C}^{\beta}\rho^{1/2-\beta}\|_2^{1/\alpha(\beta)}\leq S(\rho_{AB})+S(\rho_{BC})-S(\rho_{ABC})-S(\rho_B)\ .
 $$
 In particular, for $\beta=1/2$ the excplicit bound involving Petz's recovery map holds
 $$\left(\frac{\pi}{8}\right)^{4} \|\rho^{-1}\|^{-2}
\| \sR_\rho(\rho_B\otimes\rho_{C}) -\rho_{AB}\otimes\rho_{C}\|_1^4\leq S(\rho_{AB})+S(\rho_{BC})-S(\rho_{ABC})-S(\rho_B)\ .
 $$
 Moreover, it is clear that the equality in the strong subadditivity inequality holds if and only if Petz's map $\sR_\rho$ recovers state $\rho_{AB}\otimes\rho_C$ perfectly. 
\end{corollary}

\begin{corollary}[Operator strong-subadditivity]\label{cor:op-log}
Let $\cH=\cH_A\otimes\cH_B\otimes\cH_C$ and $\beta\in(0,1)$. Let $\rho_{ABC}$ be a state on $\cH$, and $\sigma_{ABC}=\sigma_{AB}\otimes I_C$ with $\sigma_{AB}$ being a state on $\cH_A\otimes\cH_B$. With $\alpha(\beta)$ and $N$ defined above, the following results hold.
\begin{enumerate}
\item\label{one} Theorem~\ref{thm:ssa} leads to
$$N\,  [\Tr_{AB}(P(\rho_{ABC},\sigma_{AB})P^*(\rho_{ABC}, \sigma_{AB}))]^{1/\alpha(\beta)}\leq \Tr_{AB}[\log\rho_{ABC}-\log\sigma_{AB}-\log\rho_{BC}+\log\sigma_B ]\rho_{ABC} \ .$$
\item When $\cH_B$ is one-dimensional, this becomes
$$N\,[\Tr_{A}(P^*(\rho_{AC},\sigma_{A})P(\rho_{AC}, \sigma_{A}))]^{1/\alpha(\beta)}\leq \Tr_{A}[\log\rho_{AC}-\log\sigma_{A}-\log\rho_{C} ]\rho_{AC} \ .$$
Note that $P(\rho_{AC},\sigma_{A})=\rho_{C}^{-\beta}\rho_{AC}^{1/2}-\sigma_{A}^\beta\rho_{AC}^{1/2-\beta}\ .$
\item And when $\sigma_A=\frac{1}{d_A}I_A$ is a maximally mixed state,
$$N\, [\Tr_{A}(P(\rho_{AC},\sigma_{A})P^*(\rho_{AC}, \sigma_{A}))]^{1/\alpha(\beta)}\leq \Tr_{A}[(\log\rho_{AC})\rho_{AC}-(\log\rho_{C})\rho_{C} ]+\log(d_A)\rho_C \ ,$$
with $P(\rho_{AC},\sigma_{A})=\rho_{C}^{-\beta}\rho_{AC}^{1/2}-\frac{1}{d_A^\beta}\rho_{AC}^{1/2-\beta}\ .$
\item\label{cor-Kim} Taking  $\sigma_{AB}=\rho_{AB}$ in part \ref{one}, we obtain  
\begin{align}
&N\,  [\Tr_{AB}(P(\rho_{ABC},\rho_{AB})P^*(\rho_{ABC}, \rho_{AB}))]^{1/\alpha(\beta)}\leq \Tr_{AB}[\log\rho_{ABC}-\log\rho_{AB}-\log\rho_{BC}+\log\rho_B ]\rho_{ABC}\ ,
\end{align}
where
$$P(\rho_{ABC},\rho_{AB})=\rho_{B}^\beta\rho_{BC}^{-\beta}\rho_{ABC}^{1/2}-\rho_{AB}^\beta\rho_{ABC}^{1/2-\beta}\ .$$
\item In particular, taking $\cH_B$ to be one-dimensional in the last inequality,
\begin{align}
&N\, [\Tr_{A}(P(\rho_{A},\rho_{AC})P^*(\rho_{A}, \rho_{AC}))]^{1/\alpha(\beta)}\leq \Tr_{A}[\log\rho_{AC}-\log\rho_{A}-\log\rho_{C} ]\rho_{AC}\ ,
\end{align}
where
$$P(\rho_{AC},\rho_{A})=\rho_{C}^{-\beta}\rho_{AC}^{1/2}-\rho_{A}^\beta\rho_{AC}^{1/2-\beta}\ .$$
\item\label{eq:cor-53} Theorem~\ref{thm:ssa2} leads to
\begin{equation}
N\,[\Tr_{AB}(Q^*(\sigma_{AB},\rho_{ABC})Q(\sigma_{AB},\rho_{ABC}))]^{1/\alpha(\beta)}\leq \Tr_{AB}\rho_{AB}[-\log\rho_{ABC}+\log\sigma_{AB}+\log\rho_{BC}-\log\sigma_B ]\ . 
\end{equation}
\item When $\cH_B$ is one-dimensional in the last inequality, we obtain
$$N\,[\Tr_{A}(Q^*(\sigma_{A},\rho_{AC})Q(\sigma_{A},\rho_{AC}))]^{1/\alpha(\beta)}\leq \Tr_{A}\rho_{A}[-\log\rho_{AC}+\log\sigma_{A}]+\log\rho_{C}\ . $$
\item Taking $\rho_{AB}=\sigma_{AB}$ in part \ref{eq:cor-53}, we obtain a stronger version of Kim's inequality (\ref{eq:Kim})
$$N\,[\Tr_{AB}(Q^*(\rho_{AB},\rho_{ABC})Q(\rho_{AB},\rho_{ABC}))]^{1/\alpha(\beta)}\leq \Tr_{AB}\rho_{AB}[-\log\rho_{ABC}+\log\rho_{AB}+\log\rho_{BC}-\log\rho_B ]\ . $$
\item When  $\cH_B$ is one-dimensional in the last inequality, we obtain
$$N\,[\Tr_{A}(Q^*(\rho_{A},\rho_{AC})Q(\rho_{A},\rho_{AC}))]^{1/\alpha(\beta)}\leq\Tr_{A}\rho_{A}[-\log\rho_{AC}+\log\rho_{A}]+\log\rho_{C}\ . $$ 
\end{enumerate}
\end{corollary}

%%%%%%%%%%%%%%%%%
%%%%%%%%%%%%%%%
\section{Wigner-Yanase-Dyson-type inequalities}\label{sec:WYD}
For $p\in(-1,2)$ and $p\neq 0,1$ let us take the function 
$$f_p(x):=\frac{1}{p(1-p)}(1-x^p),$$
which is {operator} convex. The quasi-relative entropy for this function is
$$S_{f_p}^K(\rho|| \sigma)=\frac{1}{p(1-p)}\Tr(K\rho K^*-K^*\sigma^{p}K\rho^{1-p})\ .$$
From Proposition \ref{prop:pinsker} we obtain the lower bound on the quasi-relative entropy in terms of the trace distance.
\begin{corollary}[Pinsker inequality]
For a unitary $U$, $p\in(-1,2)$ and $p\neq 0,1$ and states $\rho$ and $\sigma$, we have
$$\frac{1}{p(1-p)}[1-\Tr(U^*\sigma^{p}U\rho^{1-p})]\geq \frac{1}{2}\|\rho-U^* \sigma U\|_1^2\ .$$
\end{corollary}

Note that in the above inequality, as well in all below ones, it is important to leave the factor $1/p(1-p)$ in place, as it changed sign at $p=0,1$.

For the power function $f_p$ we may explicitly calculate the power $\alpha(\beta)$. Then from Example~\ref{ex-power} we have ${\rm d}\mu_f(x) = \frac{\sin(p \pi)}{p(1-p)}\,x^p\,{\rm d}x$, for $p\in(0,1)$.
Therefore, in Corollary~\ref{cor:ssa3}, we have 
$$c=\left\{\begin{matrix}
p/2 &\text{when }\beta\leq1/2\\
p(1-\beta)/(2\beta)&\text{when }\beta\geq1/2.
\end{matrix}\right. \ ,$$
and therefore,
$$\alpha(\beta)=\left\{\begin{matrix}
\frac{\beta(1-\beta)}{1+p(1-\beta)} &\text{when }\beta\leq1/2\\
\frac{\beta(1-\beta)}{2\beta+p(1-\beta)}&\text{when }\beta\geq1/2.
\end{matrix}\right. $$
Moreover, from \cite[Corollary 5.2]{CV17-2} for $\beta\leq 1/2$, the constant $N$ is defined as
\begin{eqnarray}\label{eq:pow-KL}
&&N= ( {\|K\|}+\frac{\beta}{1-\beta}\|\Delta_{\sigma,\rho}\|)^{-\frac{p(1-\beta)+1-2\beta+2\beta^2}{\beta(1-\beta)}}2^{-\frac{p(1-\beta)+1-2\beta+2\beta^2}{\beta(1-\beta)}} \\
&&\cdot\frac{\sin(p\pi)}{\pi }\left( \frac{\pi\beta(p(1-\beta)+1-2\beta+2\beta^2)}{(1+p(1-\beta))\sin\beta\pi}\right)^{\frac{1+p(1-\beta)}{\beta(1-\beta)}}\left(\frac{p(1-\beta)+1-2\beta+2\beta^2}{2(1-\beta)}\right)^{-2} \ ,\nonumber
\end{eqnarray}
and for $\beta\geq 1/2$, 
\begin{eqnarray}\label{eq:pow-KU}
&&N=\left(\frac{1-\beta}{\beta} {\|K\|}+\|\Delta_{\sigma,\rho}\| \right)^{-\frac{2\beta^2+p(1-\beta)}{\beta(1-\beta)}}2^{-\frac{2\beta^2+p(1-\beta)}{\beta(1-\beta)}}\\
&&\cdot\frac{\sin p\pi}{\pi }\left(\frac{\pi(1-\beta)(2\beta^2+p(1-\beta))}{(2\beta+p(1-\beta))\sin\beta\pi} \right)^{\frac{2\beta+p(1-\beta)}{\beta(1-\beta)}}\left(\frac{2\beta^2+p(1-\beta)}{2\beta} \right)^{-2}  \ .\nonumber
\end{eqnarray}

Recall, from part~\ref{ex:WYD} of Example~\ref{ex:quasi}, we have that the quasi-relative entropy is the Wigner-Yanase-Dyson $p$-skew information $I_p(\rho, K)$ for $K^*=K$, i.e. 
$$S_{f_{p}}^K(\rho|| \rho)=-\frac{1}{2p(1-p)}\Tr[K, \rho^p][K, \rho^{1-p}]=-\frac{1}{p(1-p)}I_p(\rho, K)\ .$$
It was conjectured by Wigner and Yanase in \cite{WY64} that $p$-skew information $I_p(\rho, K)$ is concave as a function of a density matrix $\rho$  for a fixed $p\in(0,1)$. A more general expression
$$S_{f_p}^K(\rho|| \sigma)=\frac{1}{p(1-p)}\Tr(K\rho K^*-K^*\sigma^{p}K\rho^{1-p}),$$
shows that the concavity of WYD information follows from the joint concavity of the term $\Tr(K^*\sigma^{p}K\rho^{1-p})$, since the first term in the above expression is linear in $\rho$. The concavity of this term was shown by Lieb \cite{L73} for powers of $\rho$ and $\sigma$ that sum up to a number no greater than one.

Using Corollary~\ref{cor:conv}, we have the following strengthening of the joint concavity.
\begin{corollary}[Joint concavity of WYD information] Let $p\in(-1,2)$ and $p\neq 0,1$, and  $\rho=\sum_jp_j\rho_j$ and $\sigma=\sum_j p_j\sigma_j$ (with $p_j>0$ and $\sum_j p_j=1$). With $\alpha(\beta)$ and $M$ defined above, we have
\begin{eqnarray}\label{eq:cor-conv}
{N}\left(\sum_j p_j^{1/2} \|\sigma^{\beta}K\rho^{-\beta} \rho_j^{1/2} - \sigma_j^{\beta}K\rho_j^{1/2-\beta}\|_2\right)^{1/\alpha(\beta)}\leq \frac{1}{p(1-p)}
\left(  \Tr K^*\sigma^{p}K\rho^{1-p}-\sum_j p_j \Tr K^*\sigma_j^{p}K\rho_j^{1-p}\right)\ .
\end{eqnarray}
In particular, for $\beta=1/2,$
\begin{equation}
{N}\left(\sum_j p_j^{1/2} \|\sigma^{1/2}K\rho^{-1/2} \rho_j^{1/2} - \sigma_j^{1/2}K\|_2\right)^{4}\leq\frac{1}{p(1-p)}\left( \Tr K^*\sigma^{p}K\rho^{1-p}-\sum_j p_j \Tr K^*\sigma_j^{p}K\rho_j^{1-p}\right)\ .
\end{equation} 
\end{corollary}

From Theorem~\ref{thm:ssa2} we might obtain some interesting sharpening of the operator version of Wigner-Yanase-Dyson inequalities. In \cite{R12} Ruskai  showed that taking $f_p$  in the operator version of SSA, with $\rho_{ABC}$ and $\sigma_{AB}$ leads to
$$0\leq\frac{1}{p(1-p)}[-\Tr_{AB}\rho_{ABC}^{1-p}\sigma_{AB}^p+\Tr_B\rho_{BC}^{1-p}\sigma_B^p] \ .$$

From Corollary~\ref{cor:ssa3} we obtain the error term for the above difference:
\begin{corollary}[Operator version of WYD inequality]\label{cor:opWYD}
 Let $\cH=\cH_A\otimes\cH_B\otimes\cH_C$, $\beta\in(0,1)$, {$p\in(-1,2)$, $p\neq 0,1$} and $\rho=\rho_{ABC}$, $\sigma=\sigma_{AB}\otimes I_C$. For constants defined above, the following operator inequality holds
\begin{align}
{N}\,[\Tr_{AB}(Q^*_{ABC}(\sigma_{AB},\rho_{ABC})Q_{ABC}(\sigma_{AB},\rho_{ABC})]^{1/\alpha(\beta)}\leq \frac{1}{p(1-p)}\left[-\Tr_{AB}\rho_{ABC}^{1-p}\sigma_{AB}^p+\Tr_B\rho_{BC}^{1-p}\sigma_B^p\right],
\end{align}
where
$$Q_{ABC}(\sigma_{AB},\rho_{ABC}):=\rho_{BC}^{\beta}\sigma_{B}^{-\beta} \sigma_{AB}^{1/2} -\rho_{ABC}^{\beta}\sigma_{AB}^{1/2-\beta}\ . $$
In particular, when $\cH_B$ is one-dimensional, we have
$$N\,[\Tr_{A}(Q^*_{AC}(\sigma_{A},\rho_{AC})Q_{AC}(\sigma_{A},\rho_{AC})]^{1/\alpha(\beta)}\leq \frac{1}{p(1-p)}\left[-\Tr_{A}\rho_{AC}^{1-p}\sigma_{A}^p+\rho_{C}^{1-p} \right] \ .
 $$
\end{corollary}

%%%%%%%%%%%%%%%%
\subsection{Cauchy-Schwartz inequality}

In particular, taking $p=2$ in $f_p$ defined above, gives the following right-hand side in Corollary~\ref{cor:opWYD}
$$\Tr_{AB}(\sigma_{AB}^2\rho_{ABC}^{-1}-\sigma_{B}^2\rho_{BC}^{-1})=\Tr_{AB}(\sigma_{AB}\rho_{ABC}^{-1}\sigma_{AB})-\Tr_B(\sigma_{B}\rho_{BC}^{-1}\sigma_{B})\ . $$

In \cite{LR74} Lieb and Ruskai showed the positivity of the following operator
$$\Tr_A X_{AC}^* Q_{AC}^{-1} X_{AC}\geq X_C^*Q_C^{-1}X_C, $$
for any $X_{AC}$ and positive semi-definite $Q_{AC}$ with $\ker Q_{AC}\subseteq \ker X^*_{AC}.$

From Corollary~\ref{cor:opWYD} we obtain the following sharpening of the Cauchy-Schwarz inequality:

\begin{corollary} Let $\rho_{ABC}$ and $\sigma=\sigma_{AB}\otimes I_C$ be two states on the Hilbert space $\cH=\cH_A\otimes\cH_B\otimes\cH_C$. Let $\beta\in(0,1)$. Then
$$N\,[\Tr_{AB}(Q^*_{ABC}(\sigma_{AB},\rho_{ABC})Q_{ABC}(\sigma_{AB},\rho_{ABC})]^{1/\alpha(\beta)}\leq  \Tr_{AB}(\sigma_{AB}\rho_{ABC}^{-1}\sigma_{AB})-\Tr_B(\sigma_{B}\rho_{BC}^{-1}\sigma_{B})\ ,$$
where, recall, $Q_{ABC}(\sigma_{AB},\rho_{ABC}):=\rho_{BC}^{\beta}\sigma_{B}^{-\beta} \sigma_{AB}^{1/2} -\rho_{ABC}^{\beta}\sigma_{AB}^{1/2-\beta}\ . $

Moreover, the equality
$$  \Tr_{AB}(\sigma_{AB}\rho_{ABC}^{-1}\sigma_{AB})=\Tr_B(\sigma_{B}\rho_{BC}^{-1}\sigma_{B})$$
holds if and only if
$$
\rho_{B}^{\beta}\sigma_{BC}^{-\beta} =\rho_{AB}^{\beta}\sigma_{ABC}^{-\beta}\, ,\ \ \text{ for all }\beta\in(0,1)\ .
$$
In particular, for $\beta=1/2$, it is equivalent to the Petz's recovery condition
$$
\sR_{\sigma_{AB}}(\rho_{BC})=\rho_{ABC}\ .
$$
\end{corollary} 

The reasoning for the equality condition if is similar to the one in Proposition~\ref{prop:eq}.

\vspace{0.3in}
\textbf{Acknowledgments.} A.V. is partially supported by NSF grant DMS-1812734.

%\begin{thebibliography}{000}
%\end{thebibliography}
\end{document}